\documentclass[a4paper,11pt]{article}
\usepackage{amsmath,color}
\usepackage{graphicx}
\usepackage{amsmath}
\usepackage{amsthm}
\usepackage{amssymb}
\usepackage{float}
\usepackage[font=small,labelfont=bf]{caption}
\usepackage[T1]{fontenc}
\newcommand{\dX}{\dot{X}}

\newcommand{\dY}{\dot{Y}}
\newcommand{\dy}{\dot{y}}
\newcommand{\dZ}{\dot{Z}}
\newcommand{\dz}{\dot{z}}
\newcommand{\bZ}{\mathbf Z}
\newcommand{\bY}{\mathbf Y}

 \DeclareMathOperator{\sign}{sign} 
 
\allowdisplaybreaks

\newtheorem{theorem}{Theorem}

\newtheorem{prop}[theorem]{Proposition}

\usepackage[all]{xy}
\usepackage{epstopdf}
\usepackage[T1]{fontenc}

\begin{document}

\title{Dependent time changed processes with applications to nonlinear ocean waves.}

\author{Pierre Ailliot\textsuperscript{1}, Bernard Delyon \textsuperscript{2}, Val\'erie Monbet \textsuperscript{2,3}, Marc Prevosto \textsuperscript{4}}
\vspace{20mm}
\maketitle

{\small $^1$ \em Laboratoire de Math\'ematiques de Bretagne Atlantique, UMR 6205, Universit\'e de Brest, France}\\
{\small $^2$ \em Institut de Recherche Math\'ematiques de Rennes, UMR 6625, Universit\'e de Rennes 1, France}\\
{\small $^3$ \em INRIA/ASPI, Rennes, France}\\
{\small $^4$ \em IFREMER, France}\\

\begin{center}
\date{\today}
\end{center}

\maketitle


\begin{abstract}
Many records in environmental sciences exhibit asymmetric trajectories and there is a need for simple and tractable models which can reproduce such features. In this paper we explore an approach based on applying both a time change and a marginal transformation on Gaussian processes. The main originality of the proposed model is that the time change depends on the observed trajectory. We first show that the proposed model is stationary and ergodic and provide an explicit characterization of the stationary distribution. This result is then used to build both parametric and non-parametric estimates of the time change function whereas the estimation of the marginal transformation is based on up-crossings. Simulation results are provided to assess the quality of the estimates. The model is applied to shallow water wave data and it is shown that the fitted model is able to reproduce important statistics of the data such as its spectrum and marginal distribution which are important quantities for practical applications. An important benefit of the proposed model is its ability to reproduce the observed asymmetries between the crest and the troughs and between the front and the back  of the waves by accelerating the chronometer in the crests and in the front of the waves.

\vspace{.5cm} 

\end{abstract}

\section{Introduction}
\label{sec:intro}

Marine coastal systems are subject to loadings induced by ocean waves and long time series of wave conditions are often needed to study the performances of such systems. It is very difficult to measure wave conditions on long periods of time and thus there is a need for stochastic models which can simulate quickly long time series of wave conditions with realistic characteristics. The systems are often located near the shore where the waves are known to be non-linear. Figure~\ref{fig:wave} (bottom plot) shows a typical time series  of sea-surface elevation  in shallow water. It exhibits both different shapes for crests and troughs, with sharp crests and flat troughs,  and asymmetries between the steeper front of waves and their backs. Many other situations in environmental science or in econometric lead to study asymmetric records and there is a need for simple and tractable alternatives to Gaussian processes. An approach based on applying a time change to a Gaussian process is investigated in this paper.

\begin{figure}[!ht]
\centering
\makebox{\includegraphics[scale=.75]{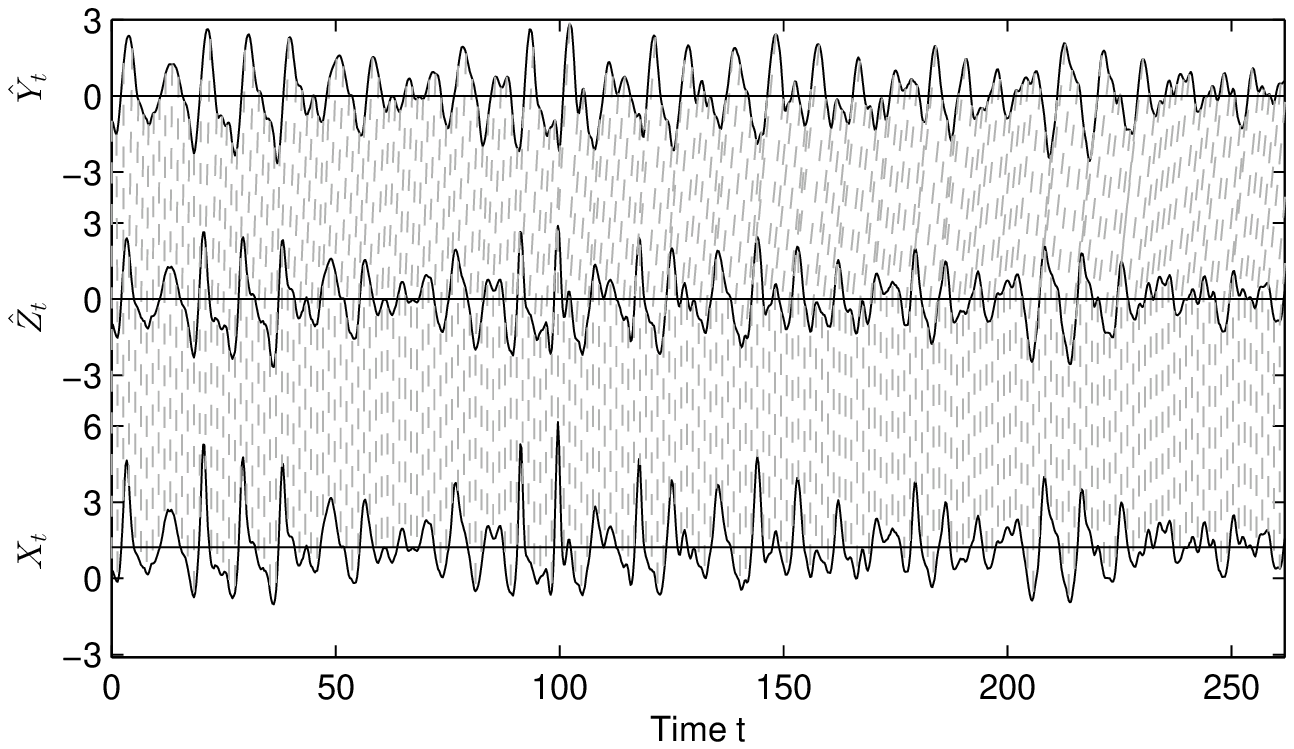}}
\makebox{\includegraphics[scale=.75]{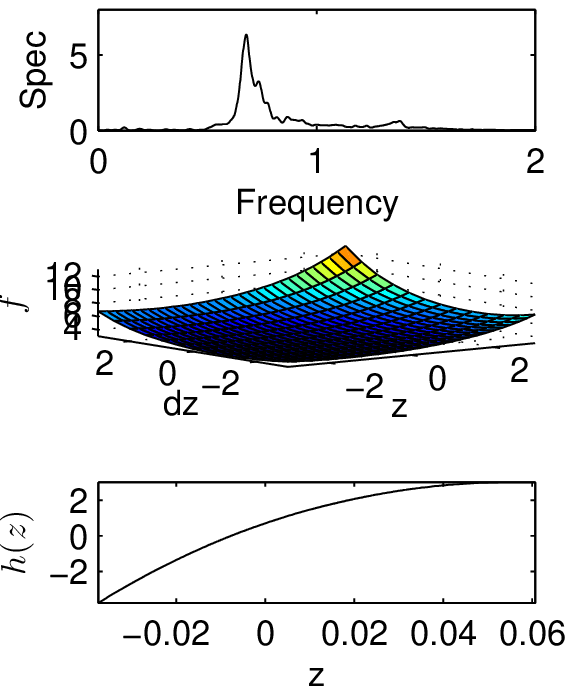}}
\caption{\label{fig:wave} {Bottom left: sequence of measured sea-surface elevation (location 7). Middle left and top left: estimated processes $\hat Z$ and $\hat Y$ for the model defined by (\ref{eq:mod2}-\ref{eq:marg}). The right panel shows the estimated functions $\hat h$ (bottom), $\hat f$ (middle) and the empirical spectral density of $Y$ (top). }}
\end{figure}

A \textit{time-changed process} $\bZ=(\bZ_t)_{t\geqslant 0}$ is defined as 
\begin{equation}
\label{eq:def}
\bZ_t = \bY_{\varphi(t)}
\end{equation}
where $\{\varphi(t)\}_{t\geqslant 0}$ denotes a non-decreasing stochastic process referred to as the \textit{time change} hereafter. $\bY=(\bY_t)_{t\geqslant 0}$ is a stochastic process which represents the evolution of the phenomenon of interest in the original time scale. Bold letters are used to stress the fact that $\bY$ and $\bZ$ are possibly multivariate processes. $\bY$ is generally referred to as the \textit{base process} in the literature.

Models with time change have become popular for financial applications  since they were introduced in  \cite{clark1973subordinated}  (see e.g. \cite{veraart2010time} for a recent review). In this application field the modified time $\varphi(t)$ is generally interpreted as a 'business time' and $\bY$ and $\varphi$ are generally assumed to be independent Levy processes.
Deformations have also been extensively used for modeling non-stationary fields in environmental applications since the seminal paper \cite{sampson1992nonparametric}. It generally consists in seeking a deterministic transformation $\varphi$ to map a stationary field $\bY$ into a non-stationary field $\bZ$ through (\ref{eq:def}).  
 These methods can deal with multidimensional (e.g. spatial) fields but they generally require  replicates of the random field as input to fit the model. More recently methods that can deal with a unique realization of a densely observed field were proposed (see \cite{anderes2008estimating} and references therein). Bayesian approaches where the deformation is modeled using a Gaussian prior were also proposed in \cite{schmidt2003bayesian} and then extended to include covariates in \cite{schmidt2011considering}.

Specific models which include a time change have already been proposed for modeling wave data. Of particular interest is the two-dimensional stochastic Gauss-Lagrange ocean wave model discussed in \cite{lindgren2009exact}. The ``space wave''  which describes the one-dimensional sea-surface elevation $Z$ at a fixed time is defined implicitly as
\begin{equation}
\nonumber
Z_{\psi(t)} = Y_t.
\end{equation}
where $\psi$ and $Y$ denote correlated Gaussian processes which describe respectively the horizontal and vertical motions. Note that the observed process is the sea surface elevation which is a univariate process.  In this model the trajectories are defined uniquely and explicitly as (\ref{eq:def}) only if $\psi$ is one to one which does not always hold true. Another model based on applying a moving average on a Laplace motion (or variance Gamma process) is used in \cite{Raillard2014} and validated on the same dataset than considered in this paper. A Laplace motion can be defined as (\ref{eq:def}) with $Y$ a Brownian motion and $\varphi$ an independent Gamma process.  
 
The originality of the model considered in this paper consists in assuming that the \textit{time change $\varphi$ is a stochastic process which depends on the base process $\bY$}. More precisely, the first model considered in this paper assumes that
\begin{equation}
\label{eq:phi1}
\bZ_t = \bY_{\varphi(t)}, ~ \varphi(t) = \int _0 ^t f(\bY_{\varphi(s)}) ds
\end{equation}
with $\bY$ a stationary (ergodic) process and $f$ a positive function referred to as the \textit{time change function} hereafter. In this model the time change is thus a function of the base process and $f$ links the speed of the chronometer to the observed process.
This model is quite general since $\bY$ can be multivariate and include for example covariates or latent variables. Note that (\ref{eq:phi1}) can be rewritten equivalently as 
\begin{equation}
\label{eq:psi1}
\bZ_{\psi(t)} = \bY_{t}, ~ \psi(t) = \int _0 ^t \frac{1}{f(\bY_{s})} ds
\end{equation}
where $\psi$ denotes the reciprocal function of $\varphi$.

Then we consider a more specific model for sea-surface wave elevation where the time change $\varphi$ is a deterministic function of the path of the process and its derivative. More precisely we assume that 
\begin{equation}
 \label{eq:mod2}
Z_t = Y_{\varphi(t)}, ~ \varphi(t) = \int _0 ^t f(Y_{\varphi(s)},\dY_{\varphi(s)}) ds
 \end{equation}
with $Y$ a univariate differentiable stationary ergodic Gaussian process,  $\dY$ the derivative of $Y$ and $f$ a positive function.
A Gaussian process $Y$ has symmetric paths and (\ref{eq:mod2}) will permit to create both \textit{crests-troughs} or \textit{vertical} asymmetries (i.e. different 
shapes for the crests and the troughs) and \textit{front-back} or \textit{horizontal} asymmetries (i.e. different shapes for the front and the back of the waves).
An example of sequence simulated with this model is shown on Figure~\ref{fig:wave}. The chosen $f$ function is increasing in both variables. It implies that the modified time $\varphi(t)$ increases quicker when the process is at high levels and thus that the crests of the time modified process $\{Z_t\}$ are narrower than the ones of $\{Y_t\}$ with the opposite holding true for the troughs. The chronometers also accelerates with the derivatives of the process and this leads to steeper wave fronts compared to the backs.

\begin{figure}[!ht]
\centering
\makebox{\includegraphics[scale=.7]{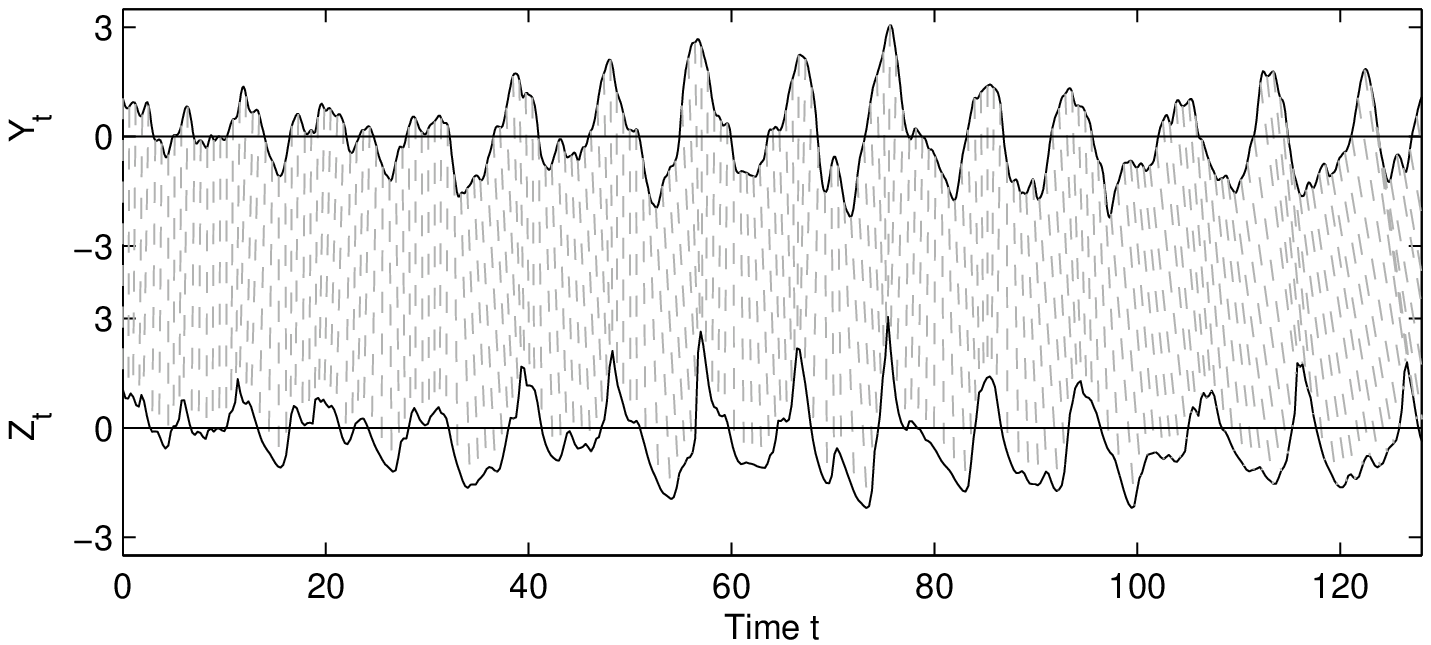}}
\makebox{\includegraphics[scale=.7]{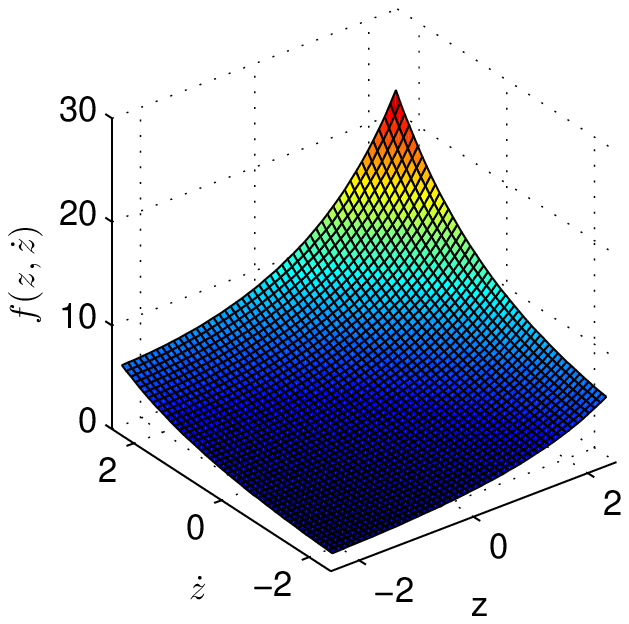}}
\caption{\label{fig:exacc} {Left: example of sequence simulated with the model (\ref{eq:mod2}). The top panel shows the Gaussian sequence and the bottom panel the time changed sequence. The vertical gray lines shows the time change. The time change function $f$ is shown on the right panel.}}
\end{figure}

The time change impacts the marginal distribution and a process defined by  (\ref{eq:mod2}) may have non-Gaussian marginal distribution. However the time change does not modify the level of the process and as a consequence the up-crossings of $Z$ and $Y$ share similar characteristics. A marginal transformation is then added to add flexibility in the modeling of the up-crossing intensity which is an important quantity for many applications. More precisely we assume that 
\begin{equation}
\label{eq:marg}
h(X_t) = Z_t
 \end{equation}
with  $Z$ defined by (\ref{eq:mod2}), $X$ the observed process and $h$  an increasing function referred to as the \textit{marginal link function} hereafter.

If $f=1$ then the process $X$ defined by (\ref{eq:mod2}-\ref{eq:marg}) satisfies $h(X_t)= Y_t$ with $\{Y_t\}$ a Gaussian process. This model, which is sometimes referred to as \textit{transformed Gaussian model} in the literature, is very classical for modeling processes or time series with non-Gaussian marginal distribution and has been extensively used in particular for wave data (see e.g. \cite{rychlik1997modelling}). 
Gaussian processes have symmetric paths and this has implications on the paths of a transformed Gaussian  process which can handle some crests-troughs asymmetries but no front-back asymmetries. This is a clear limitation of transformed Gaussian models which cannot reproduce wave sequences such as the one shown on Figure~\ref{fig:wave} which exhibits both kind of asymmetries.

The modeling principle is summarized on Figure \ref{fig:wave}. Starting from an observed trajectory with asymmetric path (bottom left), we first apply a marginal transformation  (\ref{eq:marg}) where the estimated function $h$ is shown on the bottom right panel. The concavity of this estimated function leads to decreasing the steepness of the crests compared to the troughs (middle left). Then a time deformation (\ref{eq:mod2}) is applied. The estimated $f$ function is increasing in both variables (see middle right panel) and thus the chronometer accelerates with the level and the process and its derivatives. It permits to act on both vertical and horizontal asymmetries and obtain a transformed trajectory which has approximatively symmetric path (top left) and can reasonably be  modeled using a Gaussian process. The practical motivation for this work was to develop methods to generate quickly long time series of sea-surface elevation with similar characteristics than the observed one. Using one of the methods that have been proposed in the literature to simulate Gaussian processes, realizations of $Y$ can be quickly  generated. It is then easy to deduce realizations of the transformed process $X$ once the transformation $f$ and $h$ have been estimated.

The distribution of the process $X$ defined by (\ref{eq:mod2}-\ref{eq:marg}) is characterized by the functions $f$ and $h$ and by the second order structure of the Gaussian process $Y$. This paper proposes methods to estimate these functions. For simplicity reasons we first focus on the time change model with no marginal transformation. In Section~\ref{sec:statio} we consider the general model defined by (\ref{eq:phi1}) and give conditions which imply stationarity and ergodicity. This result is then used in Sections \ref{sec:para} and \ref{sec:npara} to build respectively parametric and nonparametric estimates of $f$ in the model (\ref{eq:mod2}). Section~\ref{sec:upc} discusses the estimation of $f$ and $h$ and the second order structure of the Gaussian process $Y$ in the model defined by (\ref{eq:mod2}-\ref{eq:marg}). Section~\ref{sec:wave} discusses the results obtained when fitting this model to a specific wave dataset and conclusions are given in Section~\ref{sec:conclu}.

\section{Stationarity and ergodicity}
\label{sec:statio}

In this section we consider the model defined by (\ref{eq:psi1}). $\bZ$ and $\bY$, as well as
$\bZ_.$ and $\bY_.$ will denote the processes $(\bZ_t)_{t\geqslant 0}$ $(\bY_t)_{t\geqslant 0}$.
The stationarity of $\bY$
means that for any measurable function $g$, $g(\bY_.)$ has same distribution as $g(\bY_{.+s})$
for any $s>0$. The ergodicity means that   for any bounded measurable function $g$
\begin{align*}
\lim_{H\rightarrow\infty}\frac1H\int_0^Hg(\bY_{.+s})=E[g(\bY_.)]
\end{align*}
In order to check stationarity and ergodicity, it suffices to consider functions $g$
of the form
\begin{align}\label{gform}
g(\bY)=g_0(\bY_{t_1},\dots \bY_{t_n})
\end{align}
where $g_0$ is a bounded Borel measurable function; we may even assume that $g_0$ is
continuous. We will however keep the notation
$g(\bY_.)$ since it will appear to be handier.
\begin{prop}
\label{prop:ergo}
Let $(\bY_t)_{t\geqslant 0}$ be a  stationary process with values on $\mathbb R^d$ and continuous sample paths
and $f(.)$ be a   measurable real positive function defined on $\mathbb R^d$ such that $q_0=E\left[\frac1{f(\bY_0)}\right] < + \infty$.
Define the process $(\bZ_t)_{t\geqslant 0}$ by (\ref{eq:psi1}).
Then $\bZ$ is stationary for the measure $P'$ defined as
\begin{align}
\nonumber
&E'[g(\bZ_.)]=q_0^{-1}E\left[\frac{g(\bZ_.)}{f(\bZ_0)}\right].
\end{align}
If $\bY$ is ergodic then $\bZ$ is ergodic under $P'$. In particular for any measurable real function $g$ defined on $\mathbb R^d$ such that $E\left[\frac{|g(\bY_0)|}{f(\bY_0)}\right]<\infty$
\begin{align}
\label{eq:ergmarg}
E'[g(\bZ_t)]=\lim_{H\rightarrow\infty}\frac1H\int_0^H g(\bZ_{s}) ds=q_0^{-1}E\left[\frac{g(\bY_0)}{f(\bY_0)}\right]
\end{align}

\end{prop}

The proof of this proposition is given in the appendix. (\ref{eq:ergmarg}) is the key result for the rest of this paper. It shows that empirical means computed on a realization of the time changed process $\bZ$ converges when the length of the realization tends to infinity. The limit depends both on the distribution of the base process $\bY$ at time $0$ and on the time change function $f$.

\section{Method of moments}
\label{sec:para}

In this section we consider the model defined by (\ref{eq:mod2}). It is a particular case of the model defined by (\ref{eq:phi1}) with $\bY_t=(Y_t,\dY_t)$ and Proposition~\ref{prop:ergo} gives stationarity and ergodicity conditions for $Z$. We further assume that $Y$ is a  Gaussian process and thus for any time $t$ the random variables $Y_t$ and $\dY_t$ are independent Gaussian variables. We have $E[\dY_t]=0$ and we can assume without loss of generality that $var(\dY_t)=1$. We also assume that $E[Y_t]=0$ and $var(Y_t)=1$. It is not restrictive in the practical application discussed in Section~\ref{sec:wave} since this condition will be taken into account when estimating the marginal transformation $h$ (see Section~\ref{sec:upc}).

\subsection{Moments in the general case}
\label{sec:mom}

The following proposition gives a general expression for the joint moments of $(Z_t,\dZ_t)$.

\begin{prop}
\label{prop:mom}
Let $Y$ be a univariate differentiable ergodic stationary Gaussian process such that $E[Y_t]=0$ and $var(Y_t)=var(\dY_t)=1$. Let $U=(U_1,U_2)$ denote a two dimensional standard normal vector. Assume that (\ref{eq:mod2}) holds true with $f$ a positive measurable real function such that
\begin{align}\label{bouf2}
q_0=E\left[\frac{1}{f(U)}\right] < + \infty.
\end{align}
Then $Z$ is a differentiable ergodic stationary process.
Let $g$ be a measurable real function defined on $\mathbb R^2$ such that
 $E\left[\frac{|g(U)|}{f(U)}\right]<\infty$.
 The stationary measure $P'$ is such that
 		\begin{equation}
	\label{eq:moment}
E'[g(Z_s,\dot Z _s)] = \lim_{H\rightarrow\infty}\frac1H\int_0^H g(Z_s,\dot Z _s)ds=\frac{1}{q_0} E\left[\frac{g(U_1,U_2 f(U))} {f(U)}\right].
 	\end{equation}
\end{prop}

\begin{proof}
We have
$$\frac1H\int_0^H g(Z_s,\dot Z _s)ds = \frac1H\int_0^H g\left(Y_{\varphi(s)}, f\left(Y_{\varphi(s)},\dY_{\varphi(s)}\right) \dY_{\varphi(s)}\right)ds.$$
Now by applying Proposition~\ref{prop:ergo} with $\bY_t=(Y_t,\dY_t)$ we get that $\bZ_t=(Y_{\varphi(t)},\dY_{\varphi(t)})$ is a stationary ergodic process such that
$$ \lim_{H\rightarrow\infty}  \frac1H\int_0^H g\left(Y_{\varphi_(s)}, f\left(Y_{\varphi_(s)},\dY_{\varphi_(s)}\right) \dY_{\varphi_(s)}\right)ds =q_0^{-1}E\left[\frac{g\left(Y_{0}, f\left(Y_{0},\dY_{0}\right)\dY_0\right)}{f(Y_0,\dY_0)}\right]$$

with $ q_0=E\left[\frac1{f(Y_0,\dY_0)}\right]$. Thus (\ref{eq:moment}) holds true since $Y_0$ and $\dY_0$ are independent $\mathcal N(0,1)$ variables. 
\end{proof}

The previous proposition shows that the time change function $f$ is related to the joint distribution of $(Z_t,\dZ_t)$ and thus suggests estimating $f$ to match the empirical joint distribution of $(Z_t,\dZ_t)$. Note however that the relation is not straightforward and that it seems difficult to deduce an estimate of $f$ from an estimate of this joint distribution in the general case. Two different simplifications are considered hereafter. In the first one, which is discussed in the rest of this section, we assume a parametric shape for $f$. In the second one, which is discussed in the Section~\ref{sec:npara}, we assume that $f$ is product of two functions. This leads to an explicit expression for the joint probability density function (pdf) of $(Z_t,\dZ_t)$ and the derivation of nonparametric estimates for $f$.

\subsection{Moments in the quadratic exponential model}
 In the rest of this section we assume that the time change function $f$ has a quadratic exponential shape defined as 
\begin{align} 
\label{eq:defgauss}
 f(y,\dy) =  \exp (\beta_0+\beta_1 y + \beta_2 \dy + \frac{\beta_{1,1}}{2} y^2 + \beta_{1,2} y\dy + \frac{\beta_{2,2}}{2} \dy^2).  
\end{align}
with unknown parameter $\beta=(\beta_0,\beta_1,\beta_2 \beta_{1,1}, \beta_{1,2}, \beta_{2,2}) \in \mathbb R^6$. $\beta_1$ and $\beta_{1,1}$ describe the effect of the level of the base process on the time change and thus the vertical asymmetry whereas $\beta_2$ and $\beta_{2,2}$ describe the effect of the derivative and the horizontal asymmetry. $\beta_{1,2}$ may be interpreted as an interaction coefficient and $\beta_0$ as a scale factor. We discuss the realism of this condition for wave data in Section~\ref{sec:wave}. Other parametric  time change function may be dealt with in a similar way.

\begin{prop}
\label{prop:momgaus}
Let $Y$ be a univariate differentiable ergodic stationary Gaussian process such that $E[Y_t]=0$ and $var(Y_t)=var(\dY_t)=1$. Assume that (\ref{eq:mod2}) holds true with $f$ given by (\ref{eq:defgauss}) and that the matrix $Q=\Big(\begin{array}{cc} \beta_{1,1}+1 & \beta_{1,2} \\ \beta_{1,2} & \beta_{2,2}+1 \end{array} \Big)$ is definite positive. Then $Z$ is a differentiable ergodic stationary process. The stationary marginal distribution of $Z$ is Gaussian with mean $E'[Z_t]$ and variance $var'(Z_t)$ such that 
\begin{align} 
\label{eq:m1}
var'(Z_t)&=\frac{1}{1+\beta_{1,1}-\frac{\beta_{1,2}^2}{1+\beta_{2,2}}}\\
\label{eq:m2}
\frac{E'[Z_t]}{var'(Z_t)}&=-\beta_1+\frac{\beta_2 \beta_{1,2}}{1+\beta_{2,2}}
\end{align} 

Denote $m_1=E[|U_1|\log|U_1|]\simeq 0.046$, 
$m_2 =E[|U_1|(\log|U_1|)^2] \simeq 0.33$,
$m_3 =E[|U_1|^3 \log|U_1]]=2m_1+\sqrt{\tfrac2\pi}$. Then
\begin{align}
\label{eq:m3} 
q_0E'\left[|\dot Z_t|\right]&=\sqrt{\tfrac2\pi}\\
\label{eq:m4}
q_0E'\left[\dot Z_t\log(|\dot Z_t|)\right] &= \beta_2\\
\label{eq:m5}
q_0E'\left[Z_t\dot Z_t\log(|\dot Z_t|)\right] &= \beta_{1,2}\\
\label{eq:m6}
q_0E'\left[|\dot Z_t|\log(|\dot Z_t|)\right] &= m_1+(\beta_0+\tfrac12\beta_{1,1}+\beta_{2,2})\sqrt{\tfrac2\pi}\\
\label{eq:m7}
q_0E'\left[Z_t|\dot Z_t|\log(|\dot Z_t|)\right] &= \beta_1\sqrt{\tfrac2\pi}\\
\label{eq:m8}
q_0E'\Big[|Z_t\dot Z_t|\log(|\dot Z_t|)\Big] &=m_1\sqrt{\tfrac2\pi}+\tfrac2\pi(\beta_0+\beta_{1,1}+\beta_{2,2})\\
\label{eq:m9}
q_0E'\left[|\dot Z_t| \left(\log(|\dot Z_t|)-\beta_1 Z_t -\tfrac12\beta_{1,1} Z_t^2\right)^2\right] 
&=m_2+\beta_0^2\sqrt{\tfrac2\pi}+2\beta_{2,2}^2\sqrt{\tfrac2\pi}
+2\beta_0 m_1+\beta_{2,2}\theta_3\\
\nonumber
&~~~+2\beta_0\beta_{2,2}\sqrt{\tfrac2\pi}+2\sqrt{\tfrac2\pi}(\beta_2^2+\beta_{1,2}^2)
\end{align}
\end{prop}

\begin{proof}
If $Q$ is definite positive then condition (\ref{bouf2}) is satisfied and thus  (\ref{eq:moment}) holds true. Applying (\ref{eq:moment}) with $g(u_1,u_2)=g(u_1)$ a measurable bounded real function  we get
$$E'[g(Z_s)] = \frac{1}{ 2\pi q_0} \iint_{\mathbb R^2}  \frac{g(u_1)} {f (u_1,u_2)}\exp\left({-\frac{u_1^2+u_2^2}{2}}\right)du_1 du_2.$$
Then replacing $f$ by its expression (\ref{eq:defgauss}) and integrating of $u_2$ we obtain 
$$E'[g(Z_s)] = \int_{\mathbb R}  g(u_1) \phi(u_1) du_1$$ where $\phi$ denotes the pdf of Gaussian distribution with mean and variances given by (\ref{eq:m1}-\ref{eq:m2}).

 Now applying (\ref{eq:moment}) with $g(u_1,u_2)=|u_2|$ we get
$$E'[|\dZ_s|)] = \frac{1}{q_0} E\left[\frac{|U_2 f(U)|} {f(U)}\right]= E\left[|U_2|\right] = \sqrt{\tfrac2\pi} $$
and thus (\ref{eq:m3}). Similar calculations lead to (\ref{eq:m4}-\ref{eq:m9}).
\end{proof}

The model which defines $Z$ in Proposition~\ref{prop:momgaus} is referred to as the \textit{quadratic exponential  model} hereafter.

\subsection{Some simulation results in the quadratic exponential  model}
\label{sec:simup} 

Various estimates of $\beta$ can be built based on the relations given in Proposition \ref{prop:momgaus}. Two possibilities which we have tried are introduced in this section and compared using simulations.

A first estimate, denoted $\hat \beta$ hereafter, can be obtained by solving (\ref{eq:m4}-\ref{eq:m9}) with the theoretical moments replaced by their empirical estimates and $q_0$ estimated using (\ref{eq:m3}). In practice $\beta_2$, $\beta_{1,2}$ and $\beta_1$ are obtained using respectively (\ref{eq:m4}), (\ref{eq:m5}) and (\ref{eq:m7}). Then an expression for $\beta_{1,1}$ is derived from (\ref{eq:m6}) and (\ref{eq:m8}) (note that the term $\beta_0+\beta_{2,2}$ appears in these two formulas). Finally  $\beta_0$ and $\beta_{2,2}$ are estimated using (\ref{eq:m7}) and (\ref{eq:m9}) with the first equation used to express  $\beta_{2,2}$ as a linear function of $\beta_0$ and the second one to give a second order equation for $\beta_0$.

Simulation results indicate that $\hat \beta$ provides poor estimations for the parameters $\beta_1$ and $\beta_{1,1}$ (see Figure~\ref{fig:bootpar}) which describe the effect of the level of the process on the time change. It suggests using (\ref{eq:m1}-\ref{eq:m2}) to reestimate these parameters using
\begin{align}
\label{eq:b1}
&\beta_1=\frac{\beta_2\beta_{1,2}}{1+\beta_{2,2}}-\frac{E'[Z_t]}{var'(Z_t)}\\
\label{eq:b11}
&\beta_{1,1}=\frac1{var'(Z_t)}-1+\frac{\beta_{1,2}^2}{1+\beta_{1,2}}.
\end{align} 
with the theoretical moments being replaced by their empirical estimates and $\beta$ by $\hat \beta$ in the right hand terms. We denote $\tilde \beta=(\hat \beta_0,\tilde\beta_1,\hat \beta_2 \tilde \beta_{1,1}, \hat \beta_{1,2}, \hat \beta_{2,2}) $ the corresponding estimates.

Let us now introduce the numerical setup used to perform the simulations. We have chosen to consider parameters values obtained when fitting the model to the wave data shown on Figure~\ref{fig:wave}. The parameters $\beta_{1,2}$ was found to be close to $0$ (see Section~\ref{sec:appli} for a more detailed discussion) and it was decided to force this parameter to be equal to $0$. It permits to obtain a model with no interaction which is compatible with the assumptions made in Section~\ref{sec:npara} and allows comparison between the parametric and the nonparametric approaches. More precisely the chosen parameter values is  $\beta^*=(1.253, 0.110, 0.096, 0.068,  0 , 0.088)$. The corresponding $f$ function  is shown on  Figure~\ref{fig:exboot} and the spectral density which has been used to simulate the Gaussian process $Y$ is very close to the one shown on Figure~\ref{fig:wave}. The peak period for $Z$ is about $2.2 s$, the sampling time step is $0.027s$ and the length of the simulated sequence is $10^5$ time points. It corresponds to approximatively $1200$ waves. This is equivalent to a few hours of wave measurements in the ocean and to the length of the measurements described more precisely in Section~\ref{sec:wave}. The simulated sequences have similar characteristics than the sequence $Z$ shown on Figure~\ref{fig:wave} (left middle plot) with both crests-troughs and front-back asymmetries. 

Figure~\ref{fig:bootpar}  shows an important improvement in the estimation of $\beta_1$ and $\beta_{1,1}$ when using $\tilde \beta$ and $\hat \beta$, the other estimates being unchanged by construction. Simulations with other parameter values confirmed these results and led us to keep the estimate $\tilde \beta$ hereafter. 

\begin{figure}[!ht]
\centering
			\makebox{\includegraphics[scale=.8]{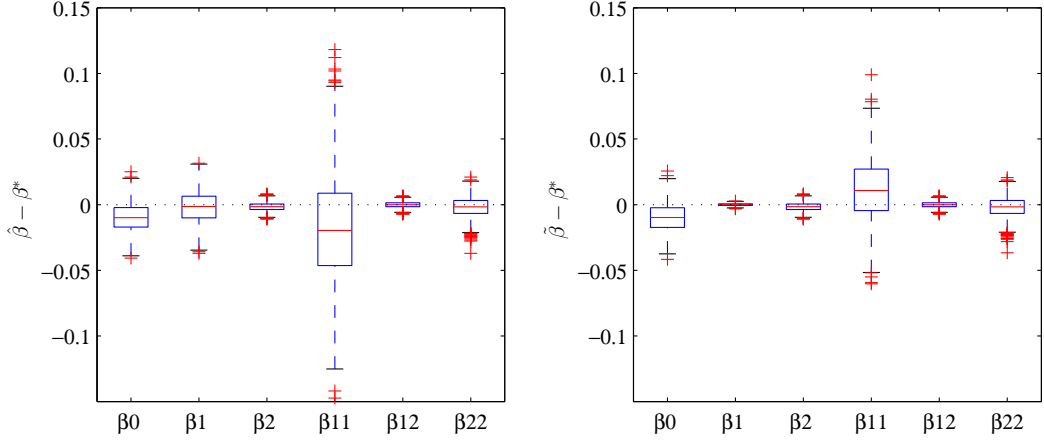}} 
\caption{\label{fig:bootpar} {Distribution of $\hat \beta -\beta^*$ (left) and $\tilde \beta -\beta^*$ (right) in the quadratic exponential model. Boxplots based on $1000$ simulated sequences. The simulation experiment is described in Section~\ref{sec:simup}.}}
\end{figure}

\section{Non-parametric estimation}
\label{sec:npara}
In this section we still consider the model defined by (\ref{eq:mod2}) with $Y$ a stationary differentiable ergodic Gaussian process. In order to get an explicit and tractable expression for the joint pdf of $(Z_t,\dZ_t)$ and deduce non-parametric estimate of $f$ we further assume that there is no interaction between the effect of the level and the derivative, i.e. for any $(y,\dy) \in \mathbb R^2$
\begin{equation}
 \label{eq:prod}
 f(y,\dy)=f_1(y)f_2(\dy).
 \end{equation}
We discuss the realism of this condition for wave data in Section~\ref{sec:wave}. In this model $f_1$ is related to the vertical asymmetries whereas $f_2$ describes horizontal asymmetries. An increasing  $f_1$ function leads to a process $Z$ with crests narrower than the troughs and  if $f_2$ is increasing then the front of the waves are steeper than the back.

The joint pdf of $(Z_t,\dZ_t)$ is given in Section~\ref{sec:jointpdf}. This result is used in Section~\ref{sec:estimnp} to build non-parametric estimates of $f_1$ and $f_2$. A simulation experiment is performed in Section~\ref{sec:simunp} to assess the performance of the estimates.
	 
\subsection{Joint distribution of $(Z_t,\dZ_t)$ in the model with no interaction}
\label{sec:jointpdf}

\begin{prop}
\label{prop:f1f2}
Let $Y$ be a univariate differentiable ergodic stationary Gaussian process such that $E[Y_t]=0$ and $var(Y_t)=var(\dY_t)=1$. Assume that (\ref{eq:mod2}) holds true with $f$ a positive  continuously differentiable real function such that (\ref{bouf2}) and (\ref{eq:prod}) hold true.  Assume further that $f_1$ and $f_2$ are continuously differentiable functions and that the function  $\dy \mapsto \dy f_2(\dy)$ is one to one and denote $k$ its inverse function. The stationary measure $P'$ is such that for any bounded measurable function $g : \mathbb R^2 \rightarrow \mathbb R$ 
 		\begin{equation}
\nonumber
E'[g(Z_s,\dot Z _s)] = \lim_{H\rightarrow\infty}\frac1H\int_0^H g(Z_s,\dot Z _s)ds=\iint_{\mathbb R ^2} g(z,\dz) p_{Z,\dZ}(z,\dz)dz d\dz
 		\end{equation}
 		with joint pdf
\begin{equation}
\label{jointf12}
p_{Z,\dZ}(z,\dz)= \frac{1}{2\pi q_0}\frac{k'\left(\frac{\dz}{f_1(z)}\right)}{f_1(z)^2 f_2\left( k\left(\frac{\dz}{f_1(z)}\right) \right)} \exp\left(-\frac{z^2+k\left(\frac{\dz}{f_1(z)}\right)^2}{2}\right)
\end{equation}

\end{prop}

\begin{proof}
According to Proposition~\ref{prop:mom} we have
$$\lim_{H\rightarrow\infty}\frac1H\int_0^H g(Z_s,\dot Z _s)ds=\frac{1}{2\pi q_0} \iint \frac{g(u_1,u_2 f(u_1,u_2))} {f(u_1,u_2)} \exp\left(-\frac{u_1^2+u_2^2}{2}\right) du_1du_2.$$
The change of variable $z=u_1$ and $\dz=u_2f(u_1,u_2)$ gives (\ref{jointf12}).

\end{proof}

The model which defines $Z$ in Proposition~\ref{prop:f1f2} is referred to as the \textit{model with no interaction} hereafter.

\subsection{Non-parametric estimation in the model with no interaction}
\label{sec:estimnp}

\textbf{Identifiability constraints.} 
In this section we assume that the conditions of Proposition~\ref{prop:f1f2} hold true and discuss the estimation of the functions $f_1$ and $f_2$.  These functions are defined up to a multiplicative constant and hereafter we assume, without loss of generality, that
\begin{equation}
\label{eq:ident2}
\int_{-\infty}^{+\infty} \frac{\exp\left( -u^2/2 \right)}{f_2(u)}du=1
\end{equation}
and thus that $q_0=\frac{1}{2\pi}\int_{-\infty}^{+\infty} \frac{\exp\left( -u^2/2 \right)}{f_1(u)}du$.  

\textbf{Marginal and conditional pdf.} 
We then deduce from (\ref{jointf12}) that the stationary marginal pdf of $Z$  and conditional pdf $\dZ_t$ given $Z_t=z$ are given respectively by
\begin{eqnarray}
\label{pdfZ}
&&p_{Z}(z) =  \int_{\mathbb R} p_{Z,\dZ}(z,\dz)d\dz =\frac{1}{2\pi q_0}\frac{\exp\left(-\frac{z^2}{2}\right)}{f_1(z)}\\
\label{eq:loiZdZ}
&&p_{\dZ|Z=z}(\dz) =  \frac{p_{Z,\dZ}(z,\dz)}{p_{Z}(z)} =\frac{k'\left(\frac{\dz}{f_1(z)}\right)k\left(\frac{\dz}{f_1(z)}\right)}{\dz} \exp\left(-\frac{k'\left(\frac{\dz}{f_1(z)}\right)^2}{2}\right).
\end{eqnarray}

Estimates of $f_1$ and $f_2$ can be derived from (\ref{pdfZ}-\ref{eq:loiZdZ}) as described below.

\textbf{Estimation of $f_1$.}

The estimate of $f_1$ considered hereafter is based on (\ref{eq:loiZdZ}) where $f_1(z)$ appears as a scale factor of the conditional distribution of $\dZ_t$ given $Z_t=z$. More precisely we have  $E'[\dZ_t|Z_t=z]=0$ and 
\begin{equation}
\label{eq:estimf1}
 E'[|\dZ_t||Z_t=z] = 2 f_1(z).
\end{equation}	
$f_1$ can be estimated by plugging in  a non-parametric estimate of the conditional expectation $E'[|\dZ_t||Z_t=z]$ in this relation. In practice we used the Nadaraya-Watson kernel estimate originally proposed in \cite{nadaraya1964estimating,watson1964smooth}.

An alternative estimate of $f_1$ is suggested by (\ref{pdfZ}) which shows that the marginal distribution of $Z$ depends only on $f_1$ and not on $f_2$. 
More precisely we have
\begin{equation}
\label{eq:estimf1pdf}
f_1(z) =  \frac{1}{2\pi q_0}\frac{\exp(-\frac{z^2}{2})}{p_{Z}(z)}
\end{equation}
and an estimate of $f_1$ may be obtained by replacing the true pdf $p_{Z}$ by a parametric or non-parametric estimate in this relation. Note however that the unknown constant $q_0$ appears in this relation and thus need to be estimated from the data. A possible estimate of $q_0$  may be obtained using the up-crossing intensity of the zeros level  $\nu^+(Z,0)$ which satisfies (see Section~\ref{sec:upc} for more details) $$\nu^+(Z,0)=\frac{1}{2\pi q_0}.$$ This estimation strategy is not further discussed hereafter for simplicity reasons but may be useful in applications where the marginal distribution of $Z_t$ (and its tails for example) is of particular importance. We found that both approaches give similar results on simulations and the link between both estimates is further discussed in Section~\ref{sec:upc} (see (\ref{eq:linkf1})).

 \textbf{Estimation of $f_2$.} Let us denote $V_t=\frac{\dZ_t}{f_1(Z_t)}$. According to (\ref{eq:loiZdZ}) $V_t$ is independent of $Z_t$ and the pdf $p_v$ of $V_t$ is
\begin{equation}
\nonumber
p_V(v) = \frac{k'(v)k(v)}{v} \exp\left(-\frac{k(v)^2}{2}\right).
\end{equation}
It means in particular that for the model with no interaction the conditional distribution of $\dZ_t$ given $Z_t=z$ depends on the level $z$ only through a scale factor $f_1(z)$. We then deduce that
$$k'(v)k(v) \exp\left(-\frac{k(v)^2}{2}\right)=v p_V(v)$$
and thus, after integration between 0 and $v$ (note that $k$ is increasing and $k(0)=0$), we get 
$$1-\exp(-k(v)^2/2)=\int_0^v up_V(u) du $$
and finally
\begin{equation}
\label{eq:f2}
k(v) =  \sign(v)\sqrt{-2 \ln \left(1- \int_0^v zp_V(z) dz\right)}.
\end{equation}
This expression is used to derive the following non-parametric estimate for $f_2$ as described below.
\begin{enumerate}
	\item Compute $\hat V_t = \frac{ \dZ_t}{\hat f_1( Z_t)}$ where $\hat f_1$ is the estimate of $f_1$.	
	\item Estimate $\int_0^v zp_V(z) dz$ using a simple Monte Carlo approach, e.g. for $v\geq0$ by $$\frac{1}{N} \sum_{t:0 \leq \hat V_t \leq v}   \hat V_t $$ where $N$ denotes the length of the observed sequence. Plug this estimate in (\ref{eq:f2}) to deduce an estimate $\hat k$ of $k$.
	\item By definition, $k$ is the inverse function of $\dz \mapsto \dz f_2(\dz)$ and thus $$f_2\left(k(\dz)\right)=\frac{\dz}{k(\dz)}.$$ The next step consists in inverting numerically this relation after replacing $k$ by $\hat k$ to deduce an estimate $\hat f_2$ of $f_2$. 
	\item In practice  the estimate is noisy close to the origin (the computation of $\frac{\dz}{k(\dz)}$ leads to a division by $0$ at the origin) and we found it useful to smooth the estimate close to the origin. In this work  the LOESS (locally weighted polynomial regression) smoothing method was used. 
		\end{enumerate}

\subsection{Some simulation results}
\label{sec:simunp}

The performance of the estimates discussed in Section~\ref{sec:estimnp} was assessed using simulations of the quadratic exponential model with no interaction considered in Section~\ref{sec:simup}. Note that for this model we have $$\dy f_2(\dy) \propto \dy \exp\left(\beta_2 \dy + \frac{\beta_{2,2}}{2} \dy^2\right)$$ and that this function in strictly increasing on $\mathbb R$ if and only if $\beta_2^2-4\beta_{2,2}<0$. This condition is satisfied for the parameters values used for the simulations (see Section~\ref{sec:simup}). It implies that the function $\dy \mapsto \dy f_2(\dy)$ is one to one and thus Proposition~\ref{prop:f1f2} applies.

Figure~\ref{fig:exboot} shows typical parametric and non-parametric estimates of $f$ obtained on a simulated sequence. Both estimates seem to be able to reproduce the global shape of $f$ with the non-parametric estimate being more noisy. Figure~\ref{fig:bootparnpar} shows fluctuation intervals based on $1000$ simulated sequences and permits a more systematic comparison of both estimates. Again both approaches provide reasonable estimates of $f_1$ and $f_2$ and there is no clear superiority of one of them. The chosen estimation method may thus depend on the application. The parametric estimate is specific to the quadratic exponential model. If this model is not appropriate for a particular application then other parametric model may be developed or the non-parametric approach may be considered. Note however that the non-parametric approach is valid only for the model with no interaction and that this assumption may also be restrictive for some applications. We will see in Section~\ref{sec:wave} that the quadratic exponential model seems realistic for the particular wave dataset considered in this paper. In such situation the parametric estimate has the additional advantage of providing a description of the observed sequence with a small number of parameters.

\begin{figure}[!ht]
\centering
			\makebox{\includegraphics[scale=.8]{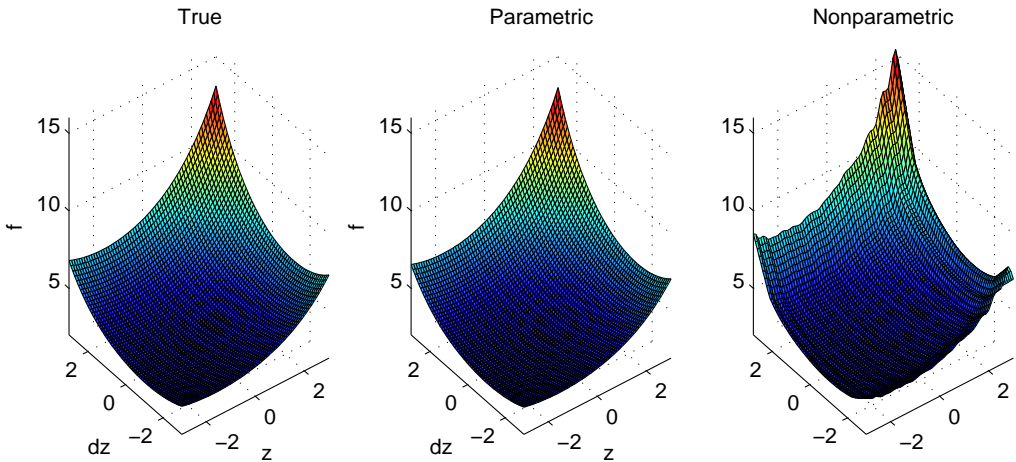}} 
\caption{\label{fig:exboot} {Comparison of the true (left) function $f$ with the parametric (middle) and non-parametric estimate (right) for a particular simulated sequence. The simulation experiment is described in Section~\ref{sec:simup}.}}
\end{figure}

\begin{figure}[!ht]
\centering
			\makebox{\includegraphics[scale=.8]{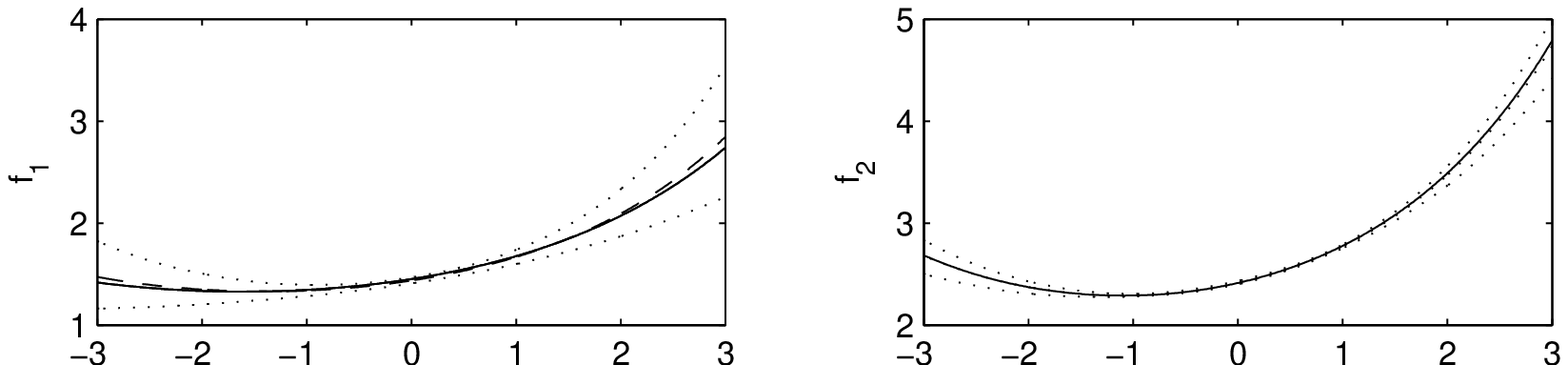}} 
					\makebox{\includegraphics[scale=.8]{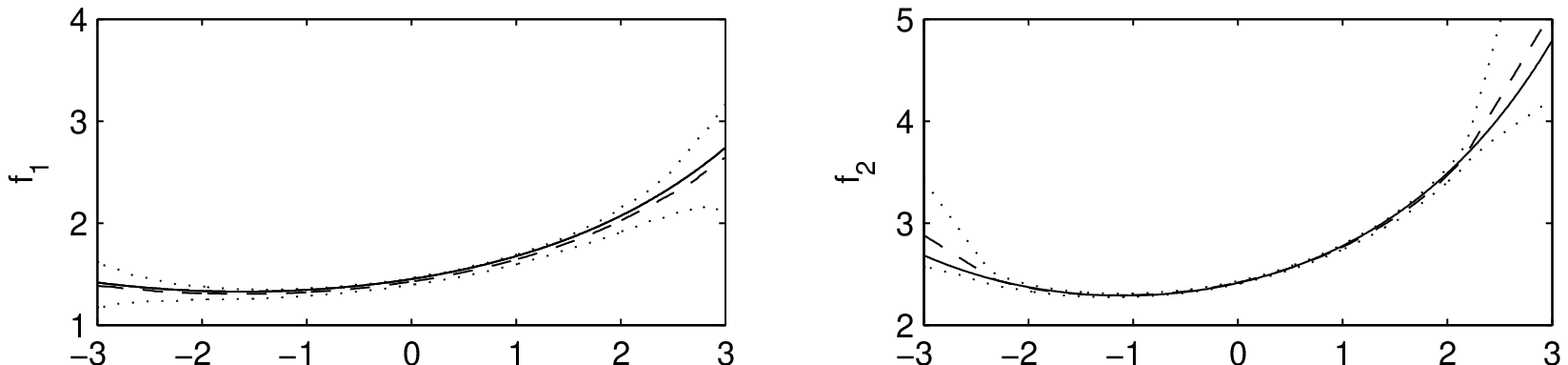}} 
\caption{\label{fig:bootparnpar} {Fluctuation intervals for the estimates of $f_1$ (left) anf $f_2$ (right) using the parametric estimate (top) and the non-parametric estimate (bottom). The full lines correspond to the true functions, the dashed line to the mean estimated values and the dotted line to 95\% fluctuation intervals.
Results based on $1000$ simulated sequences. The simulation experiment is described in Section~\ref{sec:simup}.}}
\end{figure}

\section{Model with marginal link function}
	\label{sec:upc}
	
In this section we discuss the statistical inference in the model defined by (\ref{eq:mod2}-\ref{eq:marg}) where a marginal transformation is added to the model discussed in the previous sections. Section~\ref{sec:estimh} is devoted to the estimation of the marginal link function $h$. The estimation of the time change function $f$ and of the second order structure of the Gaussian process $Y$ is discussed respectively in Sections~\ref{sec:reste} and \ref{sec:second}.

\subsection{Up-crossing intensity and estimation of $h$}
\label{sec:estimh}
As already mentioned in the introduction, if $f=1$ then the model reduces to $h(X_t)= Y_t$ with $\{Y_t\}$ a Gaussian process. This transformed Gaussian model is usual for analyzing wave data and at least two different strategies have been proposed in the literature to estimate $h$. The first one is based on the marginal distribution of $X$ which is directly related to $h$, and the second one uses the up-crossing intensity (see \cite{rychlik1997modelling,azais2009level} and references therein). 

Both approaches may be extended to the model defined by (\ref{eq:mod2}-\ref{eq:marg}). However the marginal distribution of $X$ depends on both $f$ and $h$ and it leads to complicated estimation procedures when trying to isolate the effect of $h$ on this distribution. We have thus chosen to focus on methods based on up-crossings. Indeed the time deformation does not modify the levels of the observed sequence and thus the up-crossing intensity will mainly depend on $h$ (up to a normalizing constant). This is stated more formally in the proposition below.

\begin{prop}
Let $Y$ be a univariate differentiable ergodic stationary Gaussian process such that $E[Y_t]=0$ and $var(Y_t)=var(\dY_t)=1$.  Assume that (\ref{eq:mod2}-\ref{eq:marg}) hold true with $f$ a positive measurable real function such that (\ref{bouf2}) holds true and $h$ an increasing function. 
Let $\nu^+_T(X,u)$ denotes the number of up-crossings of the level $u$ of the process $X$ on the time interval $[0,T]$  and $\nu^+(X,u) := \lim\limits_{\substack{T \to \infty}} \frac1T \nu^+_T(X,u)$ the up-crossing intensity. Then 
\begin{equation}
\label{eq:rice}
\nu^+(X,u) = \nu^+(Z,h(u))= \frac{1}{q_0}\nu^+(Y,h(u))=\frac{1}{2 \pi q_0 } \exp\left(-\frac{h(u)^2}{2}\right).
\end{equation}
\end{prop}

\begin{proof}
Since $h$ is increasing and the time change does not modify the observed level we have $\nu^+_T(X,u) = \nu^+_T(Z,h(u)) =\nu^+_{\varphi(T)}(Y,h(u))$. We deduce that

\begin{align*}
\lim\limits_{\substack{T \to \infty}} \frac{1}{T} \nu^+_T(X,u)  &= \lim\limits_{\substack{T \to \infty}} \frac{1}{T} \nu^+_T(Z,h(u))\\
&= \lim\limits_{\substack{T \to \infty}} \frac{\varphi(T)}{T} \frac{1}{\varphi(T)}\nu^+_{\varphi(T)}(Y,h(u))
\end{align*}
where $ \lim\limits_{\substack{T \to \infty}} \frac{\varphi(T)}{T}=q_0^{-1}$ (see proof of Proposition \ref{prop:ergo}). 
The result can then be deduced from the following classical result which holds true for differentiable Gaussian processes (Rice formula, see e.g. \cite{lindgren2006lectures})
 $$\lim\limits_{\substack{T \to \infty}} \frac{1}{T}\nu^+_{T}(Y,u) = \frac{1}{2 \pi} \exp\left(-\frac{u^2}{2}\right).$$
\end{proof}

The previous proposition has been used to build non-parametric estimates of $h$ as described below.  According to (\ref{eq:rice})  we get 
\begin{equation}
\label{eq:estimh2}
-\frac{h(x)^2}{2}= \log(\nu^+(X,x) )  -\log(\frac{1}{2\pi q_0}).
\end{equation}
$q_0$ needs then to be estimated from the data 	in order to be able to use this relation to estimate $h^2$. 
For that we used that
\begin{equation}
\label{eq:estimq0}
\sup_{x}\{\log\left(\nu^+(X,x)\right)\} =\log(\frac{1}{2\pi q_0}).
\end{equation} 
which shows that $q_0$ can be estimated from the mode of the up-crossing intensity (ie the intensity of the most often crossed level). Once $q_0$ is estimated we are faced to a similar estimation problem than discussed in \cite{rychlik1997modelling} for transformed Gaussian models. It consists in taking a square root of (\ref{eq:estimh2}) using the monotony of $h$ and eventually smoothing the obtained estimates (the loess method was used in this work). An isotonic regression method may also be applied to force the obtained estimates to be increasing but this was not necessary for the practical cases considered in this paper.

Figure~\ref{fig:booth} shows simulation results obtained with the setting described in Section~\ref{sec:simup} and $h$ estimated from the data (the true $h$ function used for the simulations is shown as a full line on the left plot). The bias and variance of the estimate are generally low with higher variances in the tails where the number of up-crossings is lower. Note also that the concavity of $h$ implies that the marginal distribution of $X$ has a fatter upper tail compared to the lower tail and this has consequences on the estimate of $h$ which has a higher variance in the upper tail.

\subsection{Estimation of $f$}
\label{sec:reste}

The estimation of $f$ when the process $Z$ is observed has been discussed in the previous sections. However only $X$ is generally observed when considering the model (\ref{eq:mod2}-\ref{eq:marg}). Once $h$ is estimated it is possible to compute an approximation of $Z$ as
$$\hat Z_t = \hat h (X_t)$$
with $\hat h$ the estimate of $h$. Estimates of the time change function $f$ can then be obtained by replacing the latent process $Z$ by $\hat Z$ in the methods described in Sections~\ref{sec:para} and \ref{sec:npara}. Comparing the right plots of Figures~\ref{fig:booth} and \ref{fig:bootpar} shows that using $\hat Z$ instead of $Z$ modifies the statistical properties of the  estimates of $f$ in the parametric case (similar results were obtained with the non-parametric estimates introduced in Section~\ref{sec:npara}). It is not surprising since $h$ is estimated such that the up-crossing intensity of $\hat Z$ has a Gaussian shape (see (\ref{eq:rice})) and this up-crossing intensity is related to the joint distribution of $(\hat Z,\hat \dZ)$ via the Rice formula as follows
$$ \nu^+(\hat Z,z) = \int_{\mathbb{R}} p_{\hat Z,\hat \dZ}(z,\dz) d \dz  
 \approx \frac{1}{2 \pi q_0} \exp(-z^2/2)
 $$
Estimates of $f$ are based on this joint distribution and thus may be impacted when using $\hat Z$ instead of $Z$.
Note also that this relation implies that
\begin{equation}
\label{eq:linkf1}
p_{\hat Z} (z) E[|\hat \dZ_t| |\hat Z_t=z]  \approx \frac{1}{2 \pi q_0} \exp(-z^2/2)
\end{equation}
and may explain why the non-parametric estimates of $f_1$ based on (\ref{eq:estimf1pdf}) or (\ref{eq:estimf1}) lead to similar results (it can be shown that a similar relation holds true when replacing $p_{\hat Z} (z)$ and $E[|\hat \dZ_t| |\hat Z_t=z]$ by non-parametric kernel estimates).

\begin{figure}[!ht]
\centering
			\makebox{\includegraphics[scale=.8]{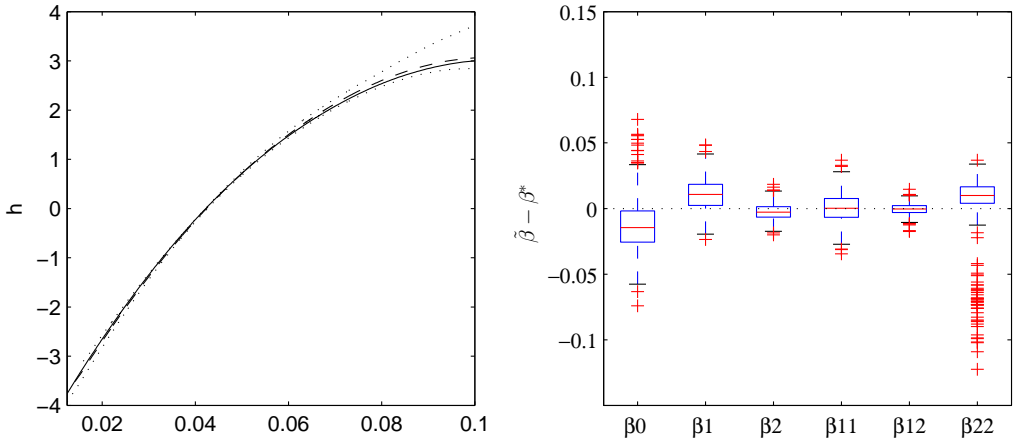}} 
\caption{\label{fig:booth} {Left plot: fluctuation interval for the estimate of $h$. The full lines correspond to the true function, the dashed line to the mean estimated values and the dotted line to 95\% fluctuation interval. Right plot: distribution of $\tilde \beta -\beta^*$ in the quadratic exponential model when $h$ is unknown. Results based on $1000$ simulated sequences. The simulation experiment is described in Section~\ref{sec:simup}.}}
\end{figure}

\subsection{Estimation of the second order structure of $Y$}
\label{sec:second}

$Y$ is stationary Gaussian process with zero mean and is thus completely characterized by its second order structure. Again the process $Y$ is not directly observed but can be approximated when the function $h$ and $f$ have been estimated. In order to simplify the presentation we assume in this section that the process $Z$ is observed. In the practical applications it is not observed but replaced by the approximation $\hat Z$ introduced in the previous section. Recall  that the time change is defined through the following differential equation
\begin{equation}
\label{edphi}
\varphi(t) = \int _0 ^t f(Y_{\varphi(s)},\dY_{\varphi(s)}) ds
\end{equation}
with $Y_{\varphi(t)}=Z_t$ and $\dY_{\varphi(s)}$ related to  $(Z_t,\dZ_t)$ implicitly as follows
\begin{equation}
\label{dYZ}
\dZ_t =  f\left(Z_t,\dY_{\varphi(t)}\right)
\end{equation}

A Euler scheme has then been used to compute an approximation of $\varphi$ defined by (\ref{edphi}) where the implicit equation (\ref{dYZ}) is solved at each time step. In the separable model defined in Section~\ref{sec:npara} it can be easily done using 
$$\dY_{\varphi(t)}= k\left(\frac{\dZ_t}{f_1\left(Z_t\right)}\right)$$
and the function $k$ is directly estimated from the data. A numerical  optimization procedure was used to solve (\ref{dYZ}) in the quadratic exponential model.
 
We can then deduce an approximation $\hat Y$ of $Y$ and compute a standard estimate of the second order structure (autocorrelation or spectral density) on the sample $\hat Y$.

\section{Application to wave data}
\label{sec:wave}

This section starts with a brief overview of physically based wave models which lead to time or space deformations. We then introduce the wave dataset and discuss the results obtained when fitting the model (\ref{eq:mod2}-\ref{eq:marg}) to this dataset.

\subsection{Physical motivations}

In the gravity wave theory the non-linear solution of the free surface evolution has been proven to be well approximated by the first order linear Lagrangian solution of the Euler equations (Gerstner wave in harmonic case, Pierson (\cite{pierson1961models}) for multi-harmonic or irregular waves).) Let $(p_0+\delta(t,p_0),\chi(t,p_0))$ denote the two-dimensional linear trajectory of a water particle of the free surface located in $p_0$ at time $t=0$, such that the height of the water surface at time $t$ and location $p_0+\delta(t,p_0)$ is $\chi(t,p_0)$. The non-linear (Eulerian) free surface elevation $\eta$ is thus given by
$$\eta(p_0+\delta(t,p_0),t)=\chi(t,p_0).$$
It corresponds to a modulation in space of the vertical movements of the particles.
In the case where the modulation is strong the solution of this implicit equation may not be unique. It corresponds in the physical meaning to wave breaking.

More recently a canonical Lie transformation was proposed for gravity waves in order to transform the real physical world in a linear (more exactly close to linear) world or inversely (see \cite{creamer1989improved}). 
In the infinite water depth case it leads to an Inviscid Burgers' equation on the free surface elevation

$$\frac{\partial}{\partial \lambda} \tilde \eta(p,t,\lambda) + \frac{1}{2} \frac{\partial}{\partial p} \tilde \eta^2(p,t,\lambda) = 0$$

with $0 \leq \lambda \leq 1$ and $\tilde \eta$ the Hilbert transform in space. $\eta_1(p,t) = \eta(p,t,1)$  and $\eta_0(p,t) = \eta(p,t,0)$ correspond to respectively 
 the free surface elevation in the linear and non-linear worlds. The solution of the Burgers' equation satisfies 
 $$\tilde \eta_0(p-\tilde \eta_1(p,t),t)=\tilde \eta_1(p,t).$$
Again modulation of space appears as a natural way to transform linear waves to non-linear ones.
In the case of finite water depth and of more complex situations with varying sea bottom in shallow
water, the solution of the Euler equations are not so simple, but the idea of modulation by
a function of the process itself could appear still pertinent.

\subsection{Data}
The proposed methodology was used to analyze wave data from an experiment carried out in a wave flume. A detailed description of the experiment can be found in \cite{becq1999non}. The bathymetric profile of the flume experiment is shown on Figure~\ref{fig:waveb} together with the various locations where the sea-surface elevation is measured (numbered from 1 close to the wave-maker to 16). The wave-maker generates random phase waves which are expected to be close to a realization of a Gaussian process with a JONSWAP-type spectra. Then the waves propagate to the right on a varying bathymetry which induces strong non-linearities. Examples of time series measured at locations $2$ and $7$ are shown on Figure~\ref{fig:waveb}. It shows in particular that both horizontal and vertical asymmetries increase when the waves propagate to the right. Free surface elevation is recorded over a duration of 40 minutes with a sampling time-step of $0.070$. The peak period is about $2.4s$. It corresponds to about 1000 waves with 34 points per wave.

\label{sec:results}
\begin{figure}[!ht]
\centering
\makebox{\includegraphics[scale=.85]{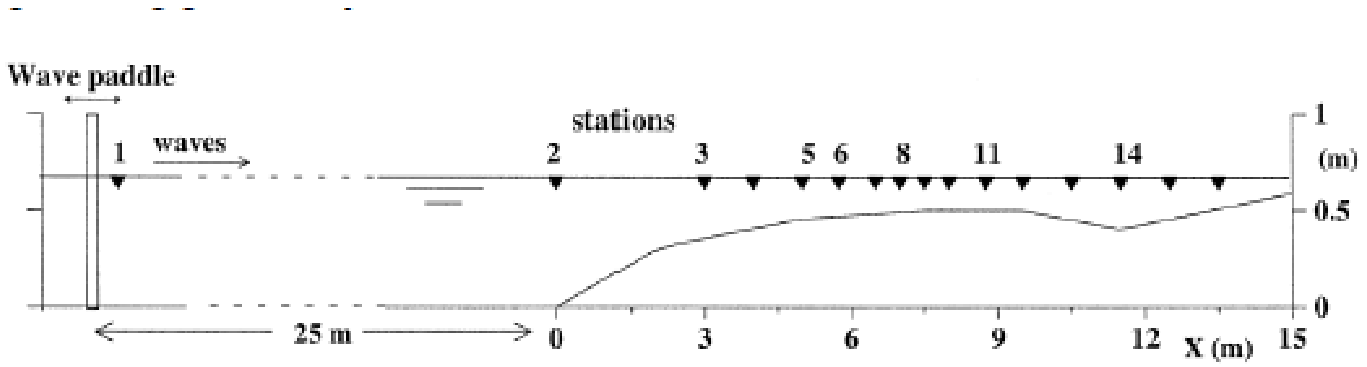}}
\makebox{\includegraphics[scale=.9]{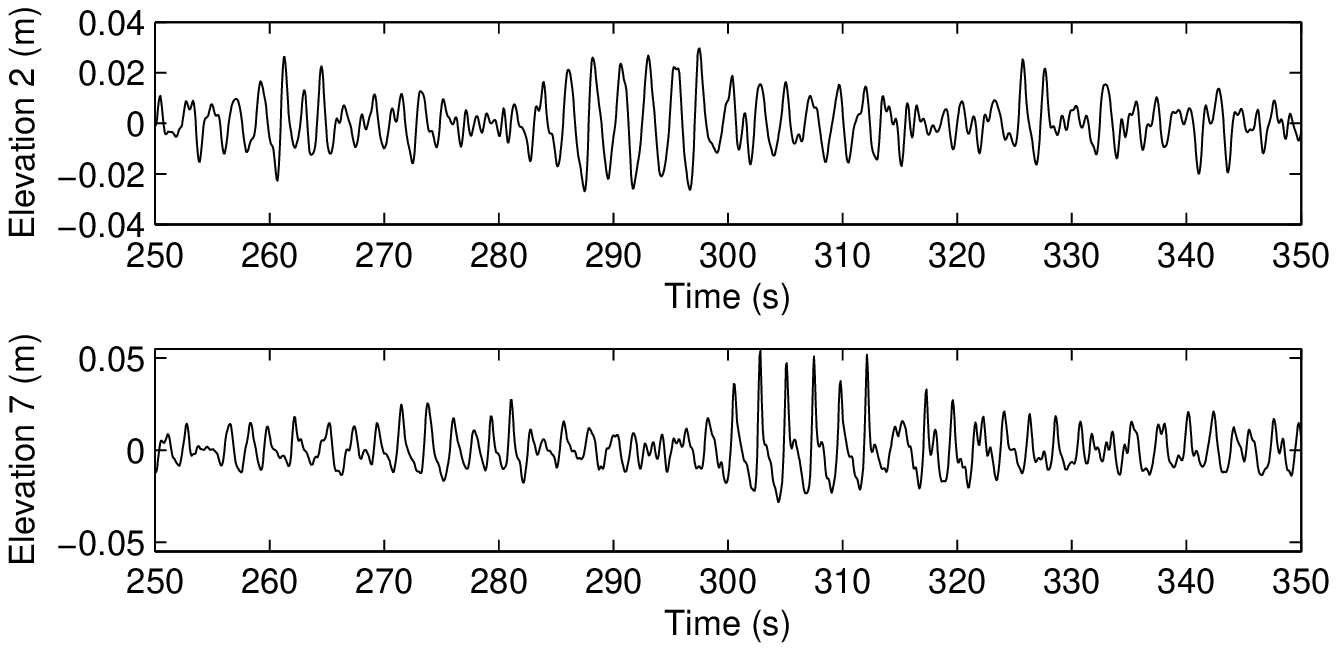}}
\caption{\label{fig:waveb} {Top panel (from \cite{becq1999non}): lay-out of the experimental set-up for the irregular flume experiment. The numbers indicate the position of the wave probes. Sequences of waves measured at stations 2 (middle panel) and 7 (bottom panel)}}
\end{figure}

\subsection{Results}
\label{sec:appli}

The model (\ref{eq:mod2}-\ref{eq:marg}) was fitted at the $16$ locations where wave measurements are available but in the discussion below we mainly focus on the central location $7$. According to Figure~\ref{fig:estimpnp7} the parametric approach for the quadratic exponential model and the non-parametric approach for the model with no interaction lead to estimates with similar shapes for $f$ at location $7$. It is a first indication that the choice of the quadratic exponential model is sensible for the wave dataset considered in this paper. 
Both the parametric and non-parametric approaches were validated using the criteria discussed below and generally slightly better results were obtained with the parametric approach. We have thus chosen to focus on this approach hereafter. It has the additional advantage of providing a parsimonious description of the observed sequence with only $6$ parameters involved. Figure~\ref{fig:evolpara} shows the evolution of the estimates with the measurement location. At locations $1$ and $2$, where it is expected that the observed sequences may be modeled as a realization of a Gaussian process, low estimated values are obtained for all the parameters except the intercept $\beta_0$. It corresponds to an almost constant time change function as in the Gaussian case. Then from location $3$ the parameters values evolve with the varying bathymetry. This evolution is relatively smooth in space. This could be used for example to interpolate the parameter values and simulate synthetic waves at locations with no observation. Note that parametric models such as power transformations and classical wave spectral models may also be considered for $h$ and the spectrum of $Y$ to get a full parametric model.

According to Figure~\ref{fig:evolpara} the values of the interaction term $\beta_{1,2}$ is always low and this suggests that a model with no interaction ($\beta_{1,2}=0$) may be appropriate. Fluctuation intervals for the estimates can be computed using similar simulations as in Section~\ref{sec:reste}. The 95\% interval  obtained for $\beta_{1,2}$ at location $7$ is $[-0.021, -0.008]$. It suggests that the interaction term is low but still statistically significant.

\begin{figure}[!ht]
\centering
\makebox{\includegraphics[scale=.8]{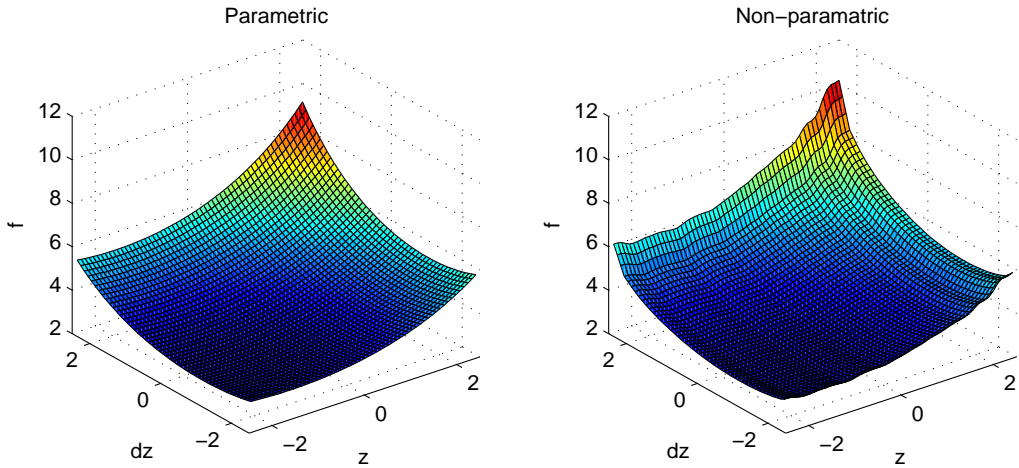}}
\caption{\label{fig:estimpnp7} {Parametric (left) and non-parametric (right) estimates of $f$ at location 7. }}
\end{figure}

\begin{figure}[!ht]
\centering
\includegraphics[scale=.8]{bassin.eps}\\
\includegraphics[scale=.64]{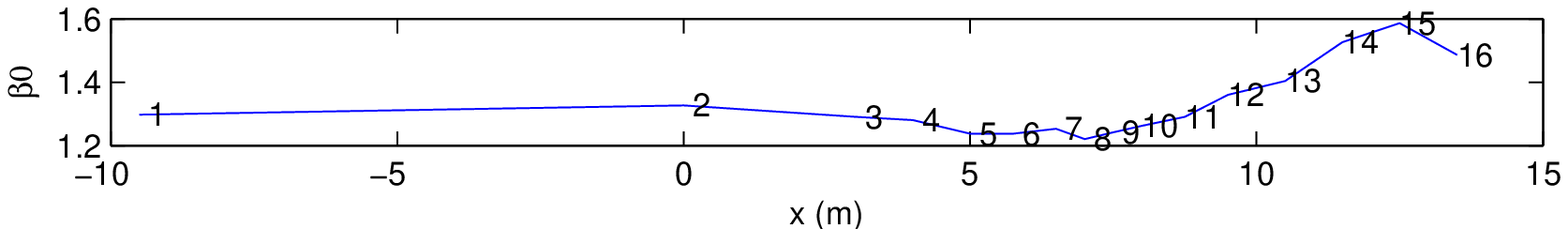} \\
\includegraphics[scale=.64]{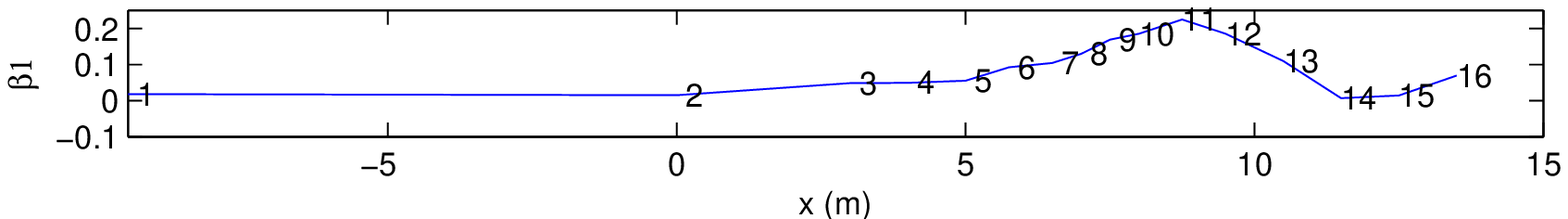} \\
\includegraphics[scale=.64]{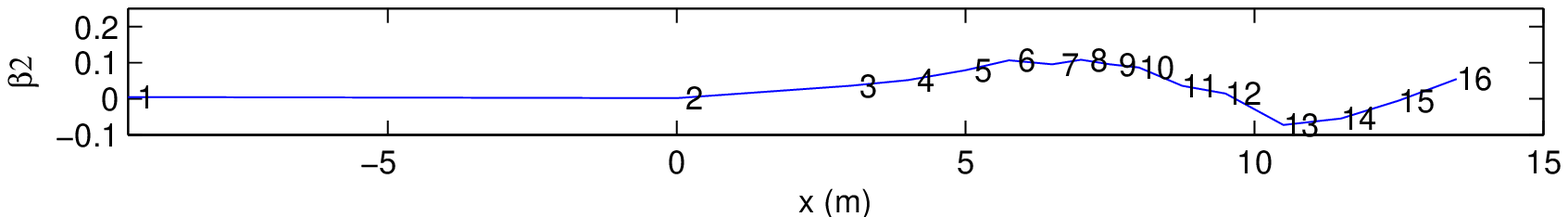}\\
\includegraphics[scale=.64]{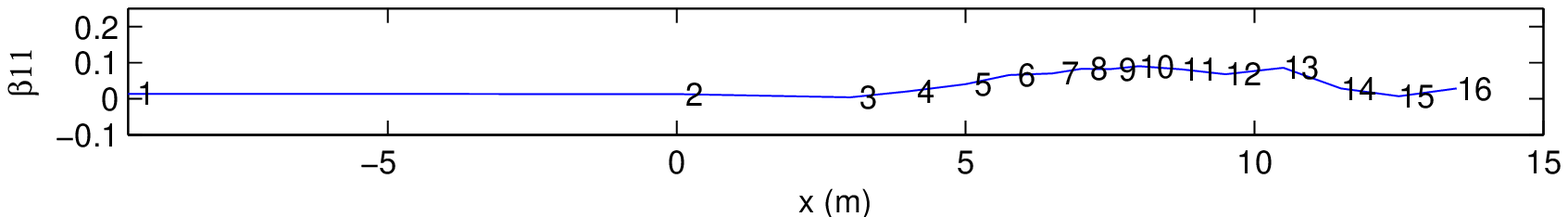}\\
\includegraphics[scale=.64]{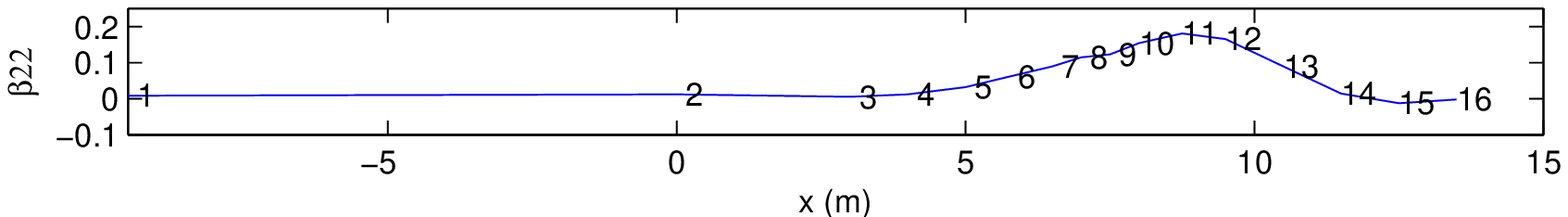}\\
\includegraphics[scale=.64]{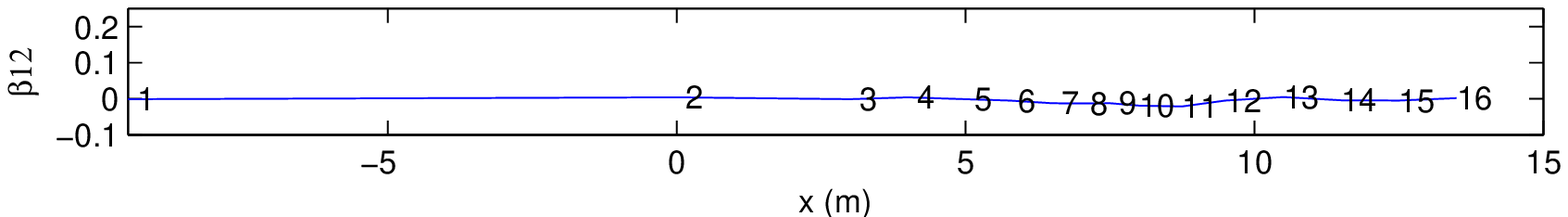}
\caption{\label{fig:evolpara} {Estimates of $\beta$ at the $16$ measurement locations.}}
\end{figure}

Let us now focus on location $7$ and validate the fitted quadratic exponential model. Figure~\ref{fig:wave} shows the estimates of $h$, $f$ and the spectral density of $Y$ at location $7$. As already mentioned in the introduction the estimated $f$ function is increasing in both variables and this permits to create both vertical and horizontal asymmetries by accelerating the chronometer in the crests and in the front of the waves. $h$ is concave and thus enhance again the sharpness of the crests by boosting high levels. If the fitted model is appropriate then the marginally and time transformed sequence $\hat Y$ should have paths which can be reasonably modeled as a realization of a Gaussian process. No formal test was performed but Figure~\ref{fig:wave} shows a sequence of this process and that the marginal and time transformations permit to symmetrize the path of the observed sequence. Another way to validate the model consists in generating sequences of the fitted model and compare the statistical properties of the simulated sequences with the ones of the original data.
We focus on this validation strategy hereafter since it is directly related to the practical application (we want a model which can generate realistic wave sequences) and it is also probably easier to interpret than validation in the Gaussian domain. Note that the fitted model can be very quickly simulated (about one half second to generate 1000 waves on a basic laptop using Matlab).  

\begin{figure}[!ht]
\centering
\makebox{\includegraphics[scale=.8]{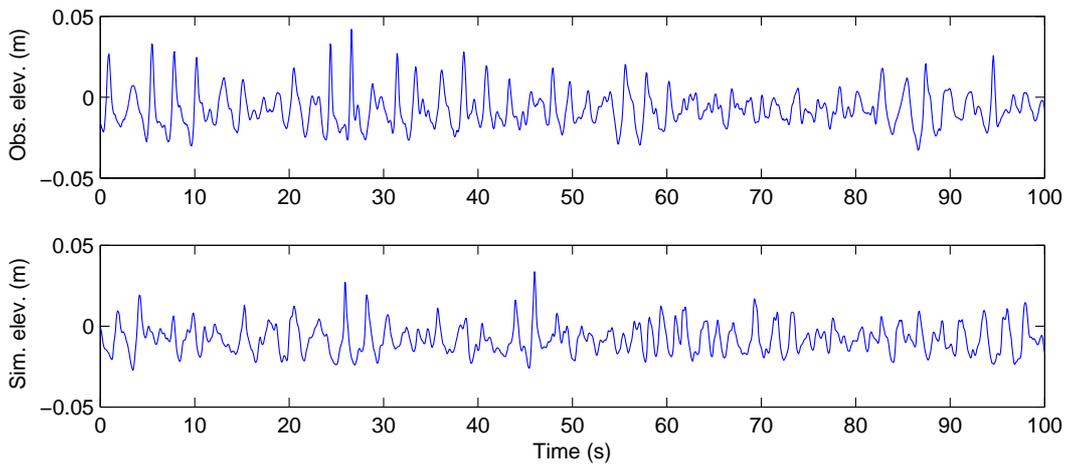}}
\caption{\label{fig:exsim} {A sequence of observation (top) and simulated data (bottom) at location 7}.}
\end{figure}

\begin{figure}[!ht]
\centering
\makebox{\includegraphics[scale=.7]{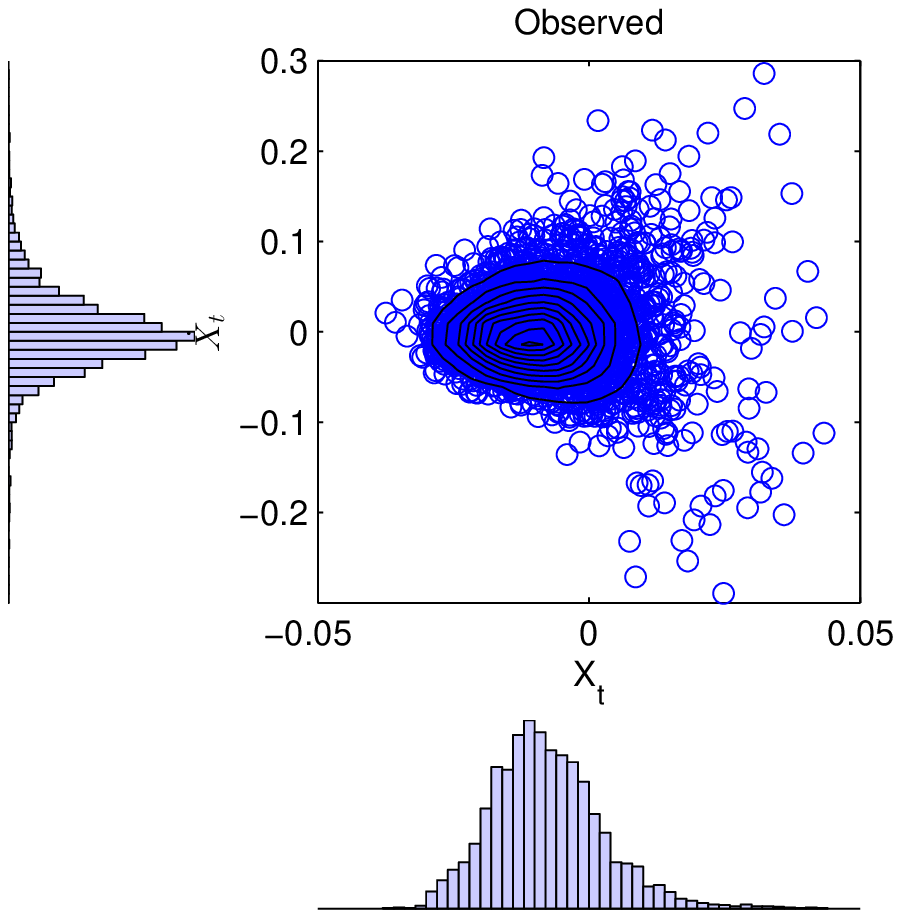}}
\makebox{\includegraphics[scale=.7]{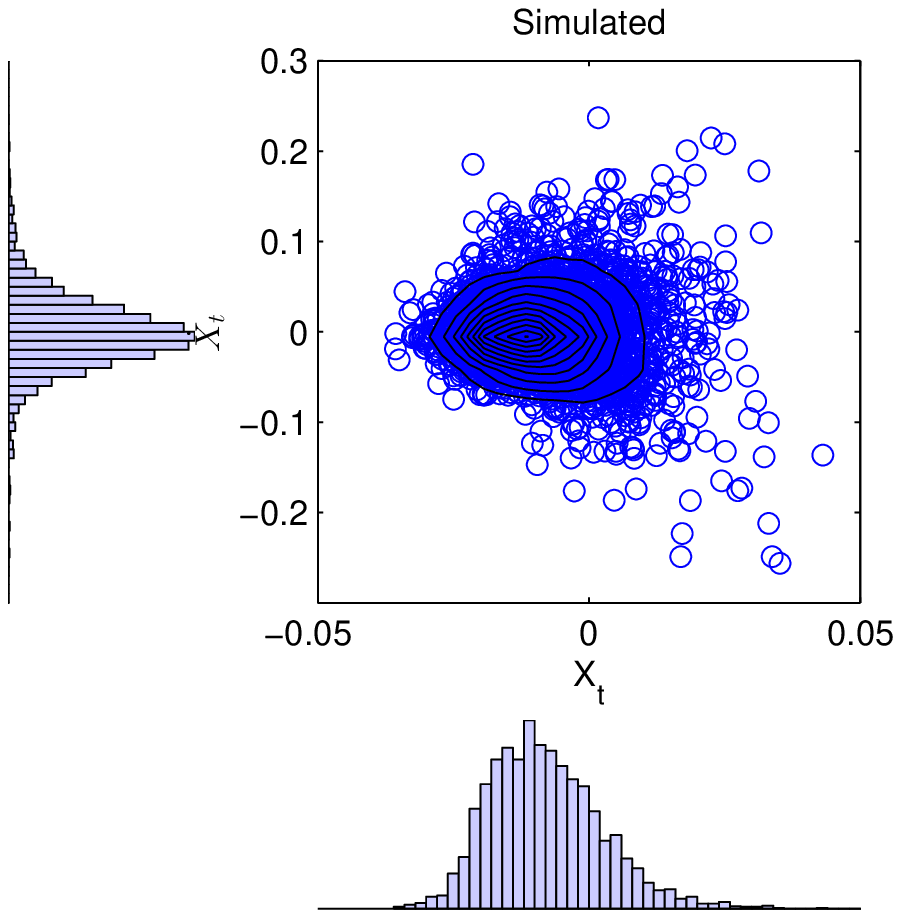}}
\caption{\label{fig:valid1} {Joint distribution of $X_t$ (x-axis) and $\dX_t$ (y-axis) for the observation (left) and the fitted model (right).}}
\end{figure}

An example of simulated sequence is shown on Figure~\ref{fig:exsim}. It looks visually coherent with the observed sequence shown on the same figure. A more systematic comparison is performed below. Figure~\ref{fig:valid1} shows the marginal and joint pdf of $(X_t,\dX_t)$. This distribution is relatively complex with an asymmetric distribution for $X_t$ (the empirical skewness of $X_t$ is $0.75$), a leptokurtic marginal distribution for $\dX_t$ (the empirical skewness of $\dX_t$ is $37.8$) and a complex relation between both variables with for example higher steepness in the crests. The fitted model is able to reproduce the complexity of this distribution. The marginal distribution of $X_t$ is of particular importance for the applications and the quantile-quantile plot shown on Figure~\ref{fig:valid2} confirms that the fit is good except maybe in the upper tail where the fit is slightly less good. An appropriate parametric model for the tails of $h$ may be useful here. 

\begin{figure}[!ht]
\centering
\makebox{\includegraphics[scale=.6]{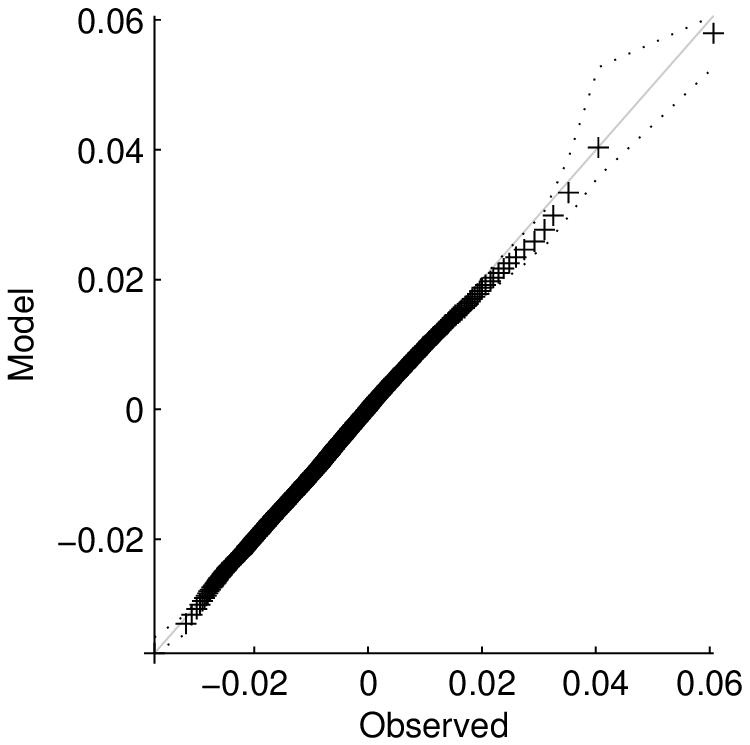}}
\makebox{\includegraphics[scale=.6]{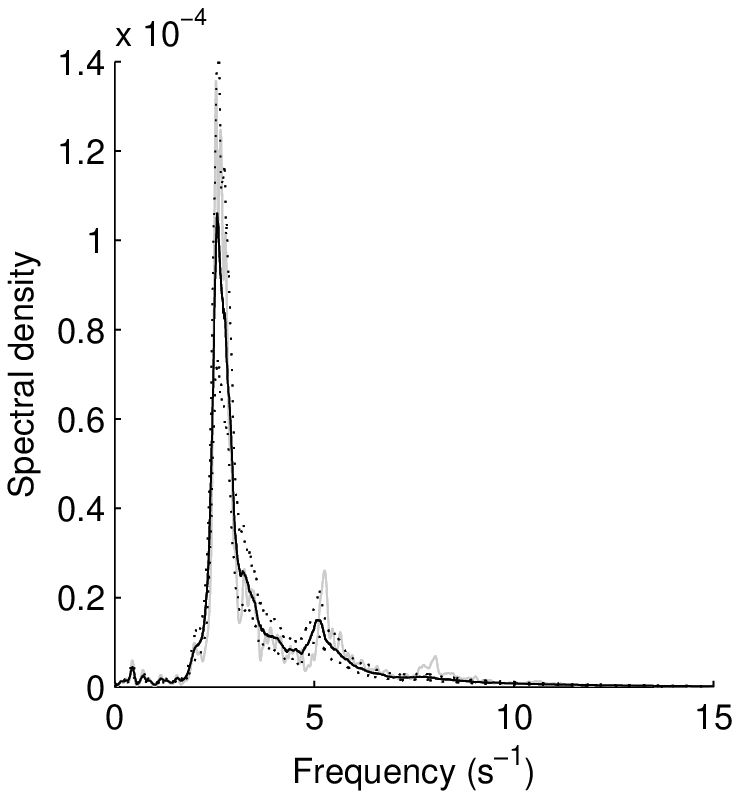}}
\makebox{\includegraphics[scale=.6]{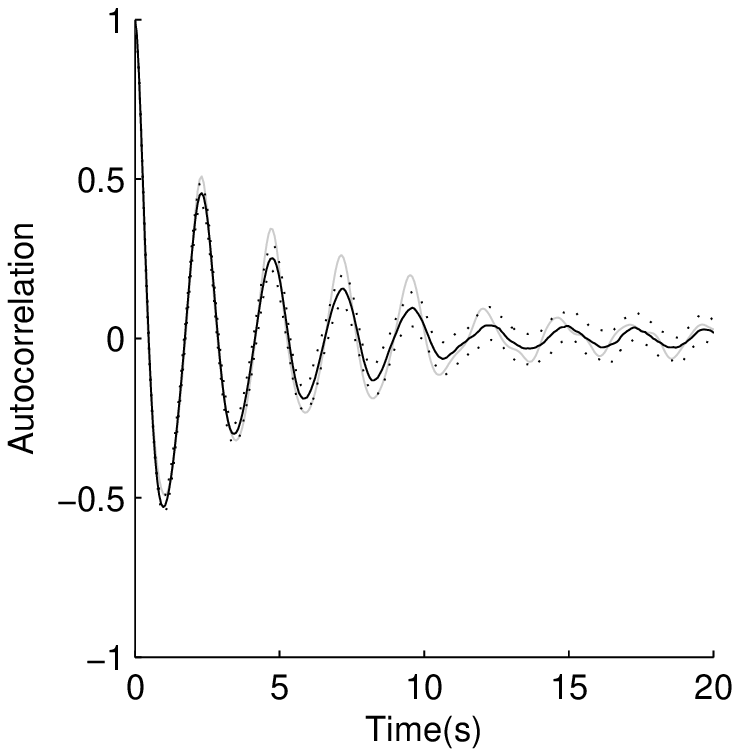}}
\caption{\label{fig:valid2} {Quantile-quantile plot, empirical spectra and autocorrelation functions at location $7$ for the data (grey curves) and the fitted quadratic exponential model (black curves). The dotted lines are 95\% fluctuation intervals computed using simulations.}}
\end{figure}

\begin{figure}[!ht]
\centering
\makebox{\includegraphics[scale=.8]{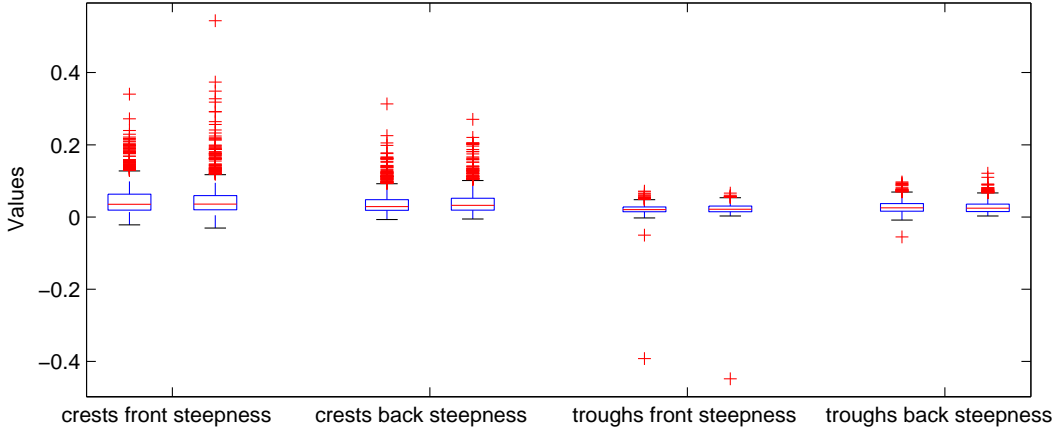}}
\caption{\label{fig:valid3} {From let to right : distribution of the the crest front steepness, crest back steepness, through front steepness, through back steepness. For each criterion, the first boxplots correspond to the observation and the second one to the fitted model. }}
\end{figure}

Figure~\ref{fig:valid2} also shows that the second order structure of the data is well reproduced by the fitted model except the third peak in the observed spectrum. This third peak is induced by non linear effect which create a low amplitude high frequency components which phase is linked to the other frequencies. The fitted model is not able to reproduce this non-linear interaction but still reproduces the associated asymmetries. Indeed Figure~\ref{fig:valid3} shows the distribution of the various slopes (crest front, crest back, trough front, trough back) computed after extracting individual waves using zeros-crossings (see \cite{Raillard2014} for a more precise definition and \cite{baxevani2014sample} for a review of existing measures of asymmetries). The differences in these four distributions are linked to the asymmetries in the sea surface elevation. For example it shows that the fronts of the crests are the steepest whereas the fronts of the troughs are the flattest. These slope distributions are well reproduced by the fitted model.

\section{Conclusions and perspectives}
\label{sec:conclu}
This paper proposes a new model for asymmetric time sequences. The originality of the proposed model consists in introducing a time change which is dependent of the observed process. The proposed statistical inference is based on the up-crossings and on the joint distribution of the process and its derivatives. Simplifying assumptions are imposed on the shape of the time change function $f$ in order to be able to derive estimates for $f$. We show that these assumptions seems realistic for the particular wave dataset considered in this paper but it may be useful to consider other cases for other applications. This remains to be further investigated together with other modeling issues such as the extension to a spatial or space-time context.

\section*{Appendix : proof of Proposition~\ref{prop:ergo}}

\begin{proof}.
Let us denote the shifted processes $\bY^\tau_.=\bY_{.+\tau}$ and $(\psi^{\tau},\bZ^\tau)$ the processes obtained by replacing $\bY$ is replaced with $\bY^\tau$ in (\ref{eq:psi1}). We get 

\begin{align*}
\psi^{\tau}(t)
=\int_0^t \frac{1}{f(\bY_s^\tau)} ds
=\int_0^t \frac{1}{f(\bY_{s+\tau})} ds
=\psi(\tau+t)-\psi(\tau)
\end{align*}
\begin{align*}
\bZ^\tau_{\psi(\tau+t)-\psi(\tau)}=\bZ^\tau_{\psi^\tau(t)} =\bY^\tau_t =\bY(t+\tau)
=\bZ_{\psi(t+\tau)}.
\end{align*}
It can then be deduced that $\bZ^\tau_.=\bZ_{.+\psi(\tau)}$ which means that the effect of a shift on $\bY$ is a random shift on $\bZ$. 

Now, by the stationarity of $\bY$, we have for any $\tau>0$ and any function $g$ of the form (\ref{gform})
\begin{align*}
q_0E'[g(\bZ_.)]=E\left[\frac{g(\bZ_.)}{f(\bZ_0)}\right]
=E\left[\frac{g(\bZ_.^\tau)}{f(\bZ^\tau_0)}\right]
=E\left[\frac{g(\bZ_{.+\psi(\tau)})}{f(\bZ_{\psi(\tau)})}\right].
\end{align*}
Hence, for any $H>0$
\begin{align*}
q_0E'[g(\bZ_.)]
&=\frac1H\int_0^HE\left[\frac{g(\bZ_{.+\psi(\tau)})}{f(\bZ_{\psi(\tau)})}\right]d\tau\\
&=E\left[\frac1H\int_0^{\psi(H)}\frac{g(\bZ_{.+s})}{f(\bZ_s)}\dot\varphi(s)ds\right]\\
&=E\left[\frac1H\int_0^{\psi(H)}g(\bZ_{.+s})ds\right]
\end{align*}
and for any $t>0$
\begin{align*}
q_0\Big(E'[g(\bZ_.)]-E'[g(\bZ_.+t)]\Big)
&=E\left[\frac1H\int_0^{\psi(H)}g(\bZ_{.+s})ds\right]
-E\left[\frac1H\int_0^{\psi(H)}g(\bZ_{.+t+s})ds\right]\\
&=E\left[\frac1H\int_0^{t}g(\bZ_{.+s})ds\right]
-E\left[\frac1H\int_{\psi(H)}^{\psi(H)+t}g(\bZ_{.+s})ds\right]\\
q_0|E'[g(\bZ_.)]-E'[g(\bZ_.+t)]|&\leqslant \frac{2t}H\|g\|_\infty.
\end{align*}
Since the right hand side can be made arbitrarily small, we have proven the stationarity.
Similar equations can be used for proving the ergodicity. Indeed we have
\begin{align}
\label{eq:ergo}
\frac1T\int_0^Tg(\bZ_{.+s})ds
&=\frac1T\int_0^{\varphi(T)}\frac{g(\bZ_{.+\psi(\tau)})}{f(\bZ_{\psi(\tau)})}d\tau
=\frac{\varphi(T)}T~\frac1{\varphi(T)}\int_0^{\varphi(T)}\frac{g(\bZ^\tau_.)}{f(\bZ^\tau_0)}d\tau
\end{align}
where $\varphi$ denotes the reciprocal function of $\psi$. Note that 
$$\lim_{T\rightarrow\infty} \frac{\psi(T)}{T}=\lim_{T\rightarrow\infty} \frac{1}{T} \int_0^{T}\frac{du}{f(\bY_u)} =q_0$$
by the ergodicity of $\bY$ and thus $\lim_{T\rightarrow\infty} \frac{\varphi(T)}{T}=q_0^{-1}$. Using again the ergodicity of $\bY$ and that $\lim_{T\rightarrow\infty}\varphi(T) = +\infty$ 
 we deduce that the second term in the right hand size term of (\ref{eq:ergo}) converges to a deterministic limit which is its expected value. \\
 (\ref{eq:ergmarg}) is a particular case of the general convergence result proved above  with $g(Y)=g_0(Y_t)$ and remarking that $Y_0=Z_0$.
 \end{proof}

\bibliographystyle{plain}
\bibliography{bib}

\end{document}